\documentclass[sigconf]{acmart}
\AtBeginDocument{%
  \providecommand\BibTeX{{%
    \normalfont B\kern-0.5em{\scshape i\kern-0.25em b}\kern-0.8em\TeX}}}

\copyrightyear{2020}
\acmYear{2020}
\acmConference[WWW '20]{Proceedings of The Web Conference 2020}{April 20--24, 2020}{Taipei, Taiwan}
\acmBooktitle{Proceedings of The Web Conference 2020 (WWW '20), April 20--24, 2020, Taipei, Taiwan}
\acmPrice{}
\usepackage{subfigure}
\usepackage[noend]{algorithmic}
\usepackage[ruled,linesnumbered,vlined]{algorithm2e}
\usepackage{enumitem}
\usepackage{multirow}
\usepackage{footmisc}
\usepackage{url}
\usepackage{hyperref}
\usepackage[hyphenbreaks]{breakurl}

\DeclareMathOperator*{\argmax}{arg\,max}

\newcommand{\overbar}[1]{\mkern 1.4mu\overline{\mkern-1.4mu#1\mkern-1.4mu}\mkern 1.4mu}

\newlength\myindent
\newtheorem{definition}{Definition}
\newtheorem{theorem}{Theorem}
\newtheorem{lemma}{Lemma}
\newtheorem{example}{Example}
\newtheorem{corollary}{Corollary}

\begin{document}
%\title{Efficient Algorithms towards Network Intervention for Improving Health Outcomes}
\title{Efficient Algorithms towards Network Intervention}

\author{Hui-Ju Hung}
\affiliation{
 \institution{The Pennsylvania State Univ., USA}
}
\email{hzh131@cse.psu.edu}

\author{Wang-Chien Lee}
\affiliation{
 \institution{The Pennsylvania State Univ., USA}
}
\email{wlee@cse.psu.edu}

\author{De-Nian Yang}
\affiliation{
 \institution{Academia Sinica, Taiwan}
}
\email{dnyang@iis.sinica.edu.tw}

\author{Chih-Ya Shen}
\affiliation{
 \institution{National Tsing Hua Univ., Taiwan}
}
\email{chihya@cs.nthu.edu.tw}

\author{Zhen Lei}
\affiliation{
 \institution{The Pennsylvania State Univ., USA}
}
\email{zlei@psu.edu}

\author{Sy-Miin Chow}
\affiliation{
 \institution{The Pennsylvania State Univ., USA}
}
\email{quc16@psu.edu}

\renewcommand{\shortauthors}{H.-J Hung et al.}
%\maketitle

\begin{abstract}
Research suggests that social relationships have substantial impacts on individuals' health outcomes. Network intervention, through careful planning, can assist a network of users to build healthy relationships.
However, most previous work is not designed to assist such planning by carefully examining and improving multiple network characteristics.
In this paper, we propose and evaluate algorithms that facilitate network intervention planning through simultaneous optimization of network \textit{degree}, \textit{closeness}, \textit{betweenness}, and \textit{local clustering coefficient}, under scenarios involving \textit{Network Intervention with Limited Degradation - for Single target (NILD-S)} and \textit{Network Intervention with Limited Degradation - for Multiple targets (NILD-M)}. We prove that NILD-S and NILD-M are NP-hard and cannot be approximated within any ratio in polynomial time unless P=NP. 
We propose the \textit{Candidate Re-selection with Preserved Dependency (CRPD)} algorithm for NILD-S, and the \textit{Objective-aware Intervention edge Selection and Adjustment (OISA)} algorithm for NILD-M. Various 
pruning strategies are designed to boost the efficiency of the proposed algorithms.
Extensive experiments on various real social networks collected from public schools and Web and an empirical study are conducted to show that CRPD and OISA outperform the baselines in both efficiency and effectiveness.
%We conduct extensive experiments on real datasets to show that CRPD and OISA outperform other baselines in terms of both efficiency and effectiveness.
\end{abstract}


\maketitle

\section{Introduction}
\label{sec:introduction}

Previous studies have shown the importance and strengths of social relationships in influencing individual behaviors. Strong social relationships have been shown to facilitate the dissemination of information, encourage innovations, and promote positive behavior~\cite{valente2012science}.
Also, social relationships surrounding individuals have substantial impacts on individuals' mental and physical health~\cite{dhand2016social,house1988social}. For example, studies in Science and American Sociological Review indicate that socially isolated individuals are more inclined to have mental health problems and physical diseases, ranging from psychiatric disorders to tuberculosis, suicide, and accidents~\cite{house1988social,durkheim1951suicide}.

To alleviate the problems that arise with social isolation, two classes of intervention strategies may be adopted, including: 1) \textit{personal intervention}, which guides individuals to understand their situations, attitudes, and capacities through counseling~\cite{cohen2000social}; and 2) \textit{network intervention}, which emphasizes the need to strengthen individuals' social networks to accelerate behavior changes that can lead to desirable outcomes at the individual, community, or organizational level~\cite{house1988social, dhand2016social}. As an example, network intervention that helps establish new social links is crucial for individuals with autism spectrum disorders~\cite{chang2016systematic}. Those new links may be effectively introduced by curative groupings and network meetings --- events which encourage them to socialize more frequently~\cite{erickson1984framework, de2011physical}.

In order for network intervention to improve individual health outcomes, it is crucial to add social links in ways that promote network characteristics found to be related to positive outcomes. Thus, an important question to ask is: given an individual, what properties define the strength of the individual's network? Previous studies point out several possibilities: 1) \textit{Degree} indicates the number of established friendships. An individual with a large degree is more popular and has more opportunities to establish self-identity and social skills~\cite{ueno2005effects}. Thus, a large-degree individual is less inclined to be socially-isolated and have mental health problems~\cite{durkheim1951suicide}.
2) \textit{Closeness} indicates the inverse of the average social distance from an individual to all others in the network. An individual with great closeness is typically located in the center of the network and tend to perceive a lower level of stress~\cite{howell2014happy}.
3) \textit{Betweenness} indicates the tendency of an individual to fall on the shortest path between pairs of other individuals. An individual with large betweenness tends to occupy \textit{brokerage} positions in the network and have more knowledge of events happening in the network. Thus, an individual with large betweenness may also perceive social relationships more accurately~\cite{lee2017adolescents}.
4) Finally, \textit{Local clustering coefficient (LCC)} 
indicates the diversity of relationships within an individual's ego network \cite{dhand2016social}.
Specifically, an individual with a small LCC tends to have diverse relationships since her friends are less likely to be acquainted with each other~\cite{dhand2016social}. The perspective that individuals with small LCCs are less likely to have mental health problems has been postulated in the functional specificity theory, which advocates the need for having different support groups for distinct functions e.g., by obtaining \textit{attachment} from families or friends, \textit{social integration} from social activity groups, and \textit{guidance} from colleagues~\cite{weiss1974provisions}. In other words, individuals with large LCCs tend not to build a diverse social network by putting all their relationship eggs in a few baskets, and tend to have depressive symptoms and neurological illnesses~\cite{dhand2016social}. 
This relationship has been validated in studies with participants across various cultures and ages.  In a study involving 173 retired US elders, higher LCC is found to be associated with lower life satisfaction, self-esteem, happiness, and higher depression~\cite{yuan2017investigating}. 
%In addition, LCC is more predictive of these outcome measures than other popular network measures that are commonly adopted as the \textit{sole target} of network intervention planning, such as closeness and betweenness. 
In another study involving 2844 high school students, higher LCC again is associated with lower self-esteem~\cite{typ}.

In practice, specialists and practitioners may not have the time or resources to provide frequent and ongoing relationship recommendations to every individual. Also, the recommendations made by persons are susceptible to their subjective biases. As such, supplemental, objective information from automated network planning algorithms that can simultaneously optimize multiple network characteristics are helpful and valuable. This paper aims to develop novel algorithms that recommend suitable intervention links based on multiple potentially health-enhancing network characteristics. However, adding social links is not always straightforward for network characteristics. Of the characteristics noted earlier, the degree for each individual can be improved by adding more edges, in which case the closeness and betweenness of the network are also enhanced. However, improving the LCC is more challenging as the LCCs of individuals and nearby friends may not always improve, and they can even deteriorate, when more edges are added.
%}
%Accordingly, we first aim to improve the LCC of individuals by adding intervention links because of the challenging neutrality of recommendation, as shown in the following example.
%\footnote{\textbf{First, adding an intervention edge changes not only the LCCs of the terminal individuals but also the LCCs of nearby individuals. In comparison, when optimizing the degree of individuals, only the terminal individuals' degrees are changed. Second, adding an intervention edge does not always improve the LCC of the terminal individuals (see the individual G in Figure~\ref{fig:ego_network}). In comparison, when optimizing the \textit{closeness} and \textit{betweenness}, adding an intervention edge only improve the closeness/betweenness of the terminal individuals.}}
%\textbf{ADD A FOOTNOTE TO DESCRIBE WHY OTHER METRICS ARE SIMPLER. BY THE WAY, it seems that you can naturally support the degree metrics since the added links are new incident link. So why don't you add a min-degree constraint?} 
%We then extend our proposed methods to incorporate other network measurements.

\begin{figure}[t]
	\centering
	\subfigure[Before] {\includegraphics[width =.99 in]{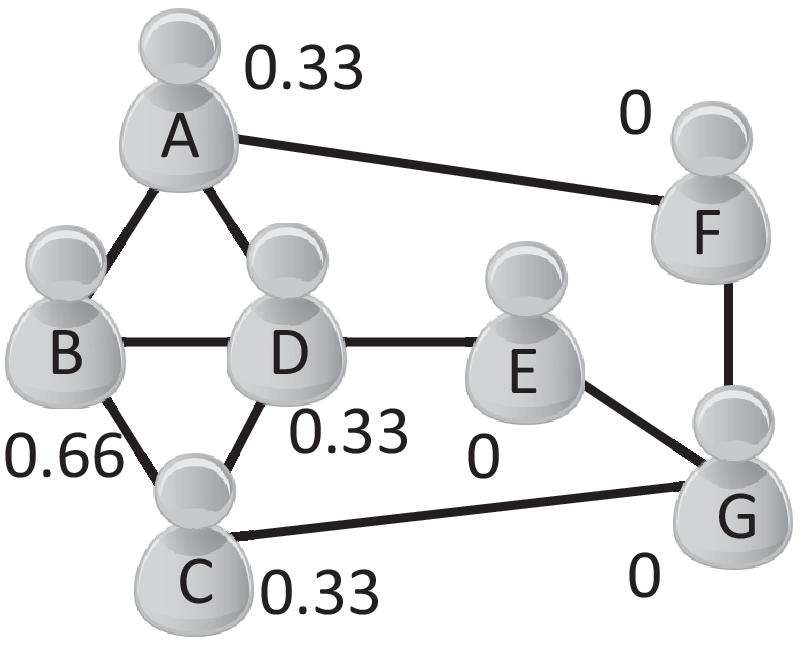}}
	\hspace{25pt}
	\subfigure[After] {\includegraphics[width =.99 in]{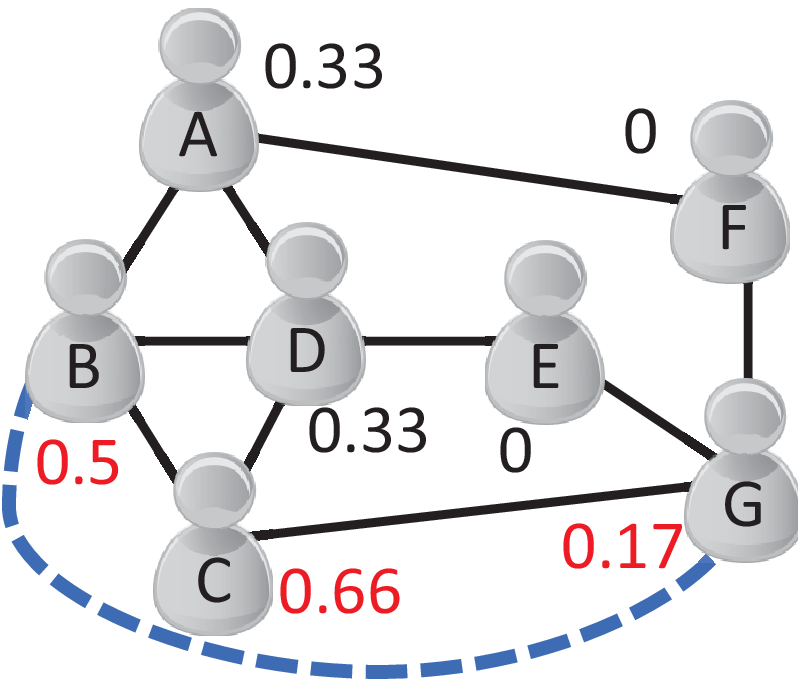}}
	\vspace{-4pt}
	\caption{An illustration on challenges of intervention}
	%\vspace{-5pt}
	\label{fig:ego_network}
\end{figure}

Selecting good intervention links based on the LCC is further deterred by other challenges. A new link established for a targeted individual may increase the LCCs of her friends, when those new friends are acquainted with each other. Figure~\ref{fig:ego_network} presents an example showing the side-effects of adding improper new links on the LCC. 
Figure~\ref{fig:ego_network}(a) shows a social network in which the nodes are annotated by their initial LCCs, and Figure~\ref{fig:ego_network}(b) presents the network after adding an edge from B to G. The LCC of B is effectively reduced to 0.5, but the LCC of C unfortunately grows to 0.66. The example shows that heuristic or uninformed selection of intervention links by specialists and practitioners may lead to undesirable changes in the LCC; even worse, the undesirable changes may happen to many nearby individuals when the network size is large. Moreover, even though a simple way to decrease LCC is to remove some existing social links, this approach is \textit{not} considered because removal of existing social links undermines established social support \cite{thoits1983multiple}.

In this paper, we propose and test several algorithms that can simultaneously optimize network LCC, closeness, betweenness, and degree.
We first formulate a new problem, namely, \textit{Network Intervention with Limited Degradation - for Single target (NILD-S)}.
Given a budget of $k$, a threshold $\tau$, and a target $t$, NILD-S finds the set $F$ of $k$ intervention edges to minimize the LCC of $t$, such that the side effect, in terms of increment in anyone's LCC, cannot exceed $\tau$. 
%Given a budget of $k$, a degradation threshold $\tau$, and a target $t$, NILD-S finds the set $F$ of $k$ intervention edges to be considered presumably by intervention specialists, such that the LCC of $t$ is minimized, and the side effect, in terms of increment in anyone's LCC, cannot exceed $\tau$.  
NILD-S also ensures that the degree, betweenness and closeness of $t$ exceed given thresholds. 
We propose \textit{Candidate Re-selection with Preserved Dependency} (CRPD) algorithm, which first obtains an initial solution by extracting the individuals with the smallest degrees, and improves the initial solution by re-examining the candidates filtered out by nodes involved in the solution. Note that CRPD selects edges according to multiple characteristics.
We prove that NILD-S is NP-hard and cannot be approximated within any ratio in polynomial time unless P=NP. %{\color{blue}
Nevertheless, we prove that CRPD can find the optimal solution for \textit{threshold graphs}, which are very similar to many well-known online social networks regarding many measurements like degree distribution, diameter and clustering coefficient~\cite{masuda2004analysis, vernitski2012astral, saha2014intergroup}.

Finally, we seek to extend the NILD-S to simultaneously improve the LCCs of multiple individuals while ensuring other network characteristics, including betweenness, closeness, and degree. To do so, we formulate the \textit{Network Intervention with Limited Degradation - for Multiple targets (NILD-M)} problem to jointly minimize the LCCs of multiple targets.  
Given the aforementioned $k$, $\tau$, and the set of targets $T$, NILD-M finds the set $F$ of $k$ intervention edges such that the maximal LCC of individuals in $T$ is minimized, while the LCC increment of any person does not exceed $\tau$. NILD-M also ensures that the degree, betweenness, and closeness of all targetes exceed their minimum thresholds.
We prove that NILD-M is NP-hard and cannot be approximated within any ratio in polynomial time unless P=NP. 
To solve NILD-M, we design \textit{Objective-aware Intervention edge Selection and Adjustment (OISA)}, which 1) carefully examines both the LCC of each terminal and the network structure to ensure the constraint of $\tau$, 2) explores the idea of \textit{optionality} to improve the solution quality, and 3) derives the lower bound on the number of required edges and the LCC upper bounds to effectively reduce computational time. 
%with the following ideas. 
%1) OISA selects the terminals of each intervention edge by carefully examining both the LCC of each terminal and the network structure, such that the changes in LCCs of other neighbors are restricted to ensure the LCC degradation constraint. 
%2) OISA ensures the degree, betweenness, and closeness for individuals to meet their minimum thresholds.
%3) OISA evaluates the reduction of LCC with different numbers of intervention edges, by deriving the \textit{lower bound} on the number of required edges to achieve any targeted LCC. It avoids pursuing a targeted LCC if the lower bound exceeds the budget $k$. 
Also, we evaluate OISA via an empirical study on four psychological outcomes, \textit{anxiety}, \textit{perceived stress}, \textit{positive and negative emotions}, and \textit{psychological well-being}.

The contributions are summarized as follows.
\vspace{-10pt}
\begin{itemize}
    \item Previous research has suggested the use of network intervention in improving health outcomes. With the potential to increase the support network by new acquaintances, however, there is no effective planning tool for practitioners to select suitable intervention edges. We formulate NILD-S and NILD-M to address this critical need for identifying suitable intervention links for a single target and a group of targets, while considering multiple network characteristics.
    
    \item We prove that NILD-S is NP-hard and cannot be approximated within any ratio in polynomial time unless P=NP. We propose CRPD and prove that CRPD obtains the optimal solution for threshold graphs.
    
    \item We prove that NILD-M is NP-hard and cannot be approximated within any ratio in polynomial time unless P=NP and design OISA for NILD-M.
   
    \item Experiments on real datasets show that the proposed CRPD and OISA efficiently find near-optimal solutions for NILD-S and NILD-M and outperform the baselines. Also, an empirical study assessed by clinical psychologists and professors in the field manifests that the network intervention alleviates self-reported health outcomes of participants, and the effects are statistically significant over another control group. 
\end{itemize}

The rest of this paper is organized as follows. Section~\ref{sec:relatedwork} reviews related work. 
Section~\ref{sec:nild} formulates NILD-S, analyzes its theoretical hardness, and proposes CRPD. Section~\ref{sec:problem} formulates NILD-M and analyzes its theoretical hardness. Section~\ref{sec:algo_oisa} proposes OISA. Section~\ref{sec:experiment} reports the experiments. Finally, Section~\ref{sec:conclusion} concludes the paper.

\section{Related Work}
\label{sec:relatedwork}

%\todo[inline]{Two problems: I think this first paragraph a bit redundant compared to earlier info in Section 1. Can you integrate them? Also, please look up John Cacioppo's work on perceived loneliness and correlates with depression and cite it below. Given the redundancy, you might want to just remove the first paragraph or move it to section 1.}
%High levels of social integration, support engagement, and engagement have been reported to be associated with lower negative emotions, and in some populations, healthy living style~\cite{berkman2000social}.
The theory of network intervention has been studied in the fields of psychology, behavioral health, and education for lowering negative emotions by enhancing social integration, support, engagement, and attachment~\cite{berkman2000social}. 
It has also been adopted for family therapy and bullying avoidance~\cite{erickson1984framework}. 
%{\color{blue} In public health, it has been adopted for information dissemination~\cite{norr2004impact}.} 
In education, network intervention has been implemented to facilitate knowledge dissemination among students, thereby improving student learning~\cite{van2002differences}. 
% Moreover, based on an analysis of mental illness in 2015 by the National Institute of Health, nearly 18.5\% of U.S. adults experienced depression or other mental disorders~\cite{nih_15}.
% %\footnote{\url{https://www.nimh.nih.gov/health/statistics/mental-illness.shtml}} 
% These disorders are the third most common cause of hospitalization, costing \$193.2 billion per year in the U.S.~\cite{insel2008assessing}. 
%Social isolation and perceived loneliness are among the multitude of risk factors that are associated with depression and mental health problems~\cite{cacioppo2006loneliness}. 
Under current practice, new intervention links are typically selected heuristically by practitioners \cite{king2006youth, sherman2009evaluation}. However, it is very challenging to consider multiple persons simultaneously without deteriorating the status of surrounding individuals. Thus, it would be worthwhile to develop algorithms for this important need.

In the field of social network analysis, researchers have paid considerable attention to efficiently finding the number of triangles~\cite{lim2015mascot, mcgregor2016better} and selecting a group of individuals with the maximum or minimum number of triangles~\cite{shen2017finding}. Notice that the above-mentioned research mostly focuses on measuring structural properties of nodes in static or dynamic networks, with no intention to tailor and change the network graph. Recently, a new line of research in network science has emerged with the objective of revising a network graph according to specific network characteristics. %{\color{blue} 
These include maximizing the closeness centrality, betweenness centrality and influence score, minimizing the diameter, %maximizing the leading eigenvalue of an adjacency matrix
and enhancing the network robustness~\cite{papagelis2015refining, crescenzi2016greedily, wilder2018optimizing}.
However, these algorithms do not include LCC as a target network characteristic for intervention purposes.  Importantly, none of these algorithms was designed to optimize multiple network characteristics simultaneously. 
%}
%5but also for different application goals (i.e., they are not improving health outomes.)
%However, these algorithms do not suit our purposes because they are not applicable for lowering the LCCs of multiple nodes while satisfying the LCC degradation constraint. 

Recently, owing to the success of online social networks, reported cases of social network mental disorders have increased,  motivating new collaborations between data scientists and mental health practitioners. 
New machine learning frameworks have been shown to be helpful in identifying patients tending to be vulnerable, and even have clinical levels of negative emotions and unhealthy living~\cite{SSY16,saha2016framework,shuai2019newsfeed, shuai2017comprehensive}.
%New machine learning frameworks have been proposed to identify potential patients with social network mental disorders, borderline personality disorder, bipolar disorder, depression disorder, and major depressive disorder~\cite{SSY16,saha2016framework}. 
%Efficient algorithms have been developed to assist online group therapy~\cite{SSY15}. 
However, those are not designed for network intervention, which actually changes the network graph.
Finally, link prediction \cite{SNM11,BBM14,wang2018shine, brochier2019link, chen2018pme, zhang2018link, fu2018link} has been widely studied. Existing algorithms usually recommend individuals sharing many common friends and similar interests to become friends. However, they are not designed for network intervention, which does not necessarily prefer people socially close or with similar backgrounds.

%Finally, link prediction \cite{SNM11,BBM14} has been widely studied. Existing algorithms usually recommend individuals sharing many common friends and similar interests to become friends. However, these algorithms are not designed for network intervention, which does not necessarily prefer people socially close or with similar background. %Even worse, the links recommended by those algorithms may increase the LCC and reduce the diversity of the individual's ego network. In contrast to conventional friend recommendations, network intervention requires more resources to build friendships with the help of trained practitioners~\cite{erickson1984framework}. Thus, more careful planning is needed.

%Finally, link prediction \cite{SNM11,BBM14} has been widely studied. Existing algorithms usually recommend individuals that share many common friends and similar interests (or profiles) to become friends. However, these algorithms are not designed for network intervention (which does not necessarily prefer people socially close or with similar background), because it may increase the LCC (due to many common friends) and even reduce the diversity of the user's ego network (not appropriate for intervention as explained in Section \ref{sec:introduction}). In contrast to conventional friend recommendations, network intervention requires more resources to build friendships with the help of trained practitioners \cite{erickson1984framework}. Therefore, it needs careful planning as suggested in this paper.

\begin{table}[t]
\caption{Summary of notations}
\scriptsize
\begin{tabular}{|ll||ll|}
\hline
Term & Meaning  & Term & Meaning              \\ \hline
$G$         & original network           & $F$   & selected intervention edges    \\ 
$t$, $T$  & targeted individual(s)     & $\overbar{G}$   & network after intervention      \\ 
$N_G(v)$    & $v$'s neighbors in $G$     & $u$   & node to be connected with $t$  \\ 
$b_G(v)$    & $v$'s betweenness in $G$   & $\hat{d}$   & maximum degree in $G$   \\ 
$c_G(v)$    & $v$'s closeness in $G$     & $R$       & nodes removed by CRPD baseline  \\
$d_G(v)$    & $v$'s degree in $G$        & $r$       & a removed node in $R$ \\ 
$LCC_G(v)$  & $v$'s LCC in $G$           & $w_b$       & weights of betweenness in MISS                   \\ 
$\omega_b$  & lower bound of betweenness & $w_c$       & weights of closeness in MISS \\ 
$\omega_c$  & lower bound of closeness   & $w_d$       & weights of degree in MISS\\ 
$\omega_d$  & lower bound of degree      & $f(\cdot)$  & weight func. in threshold graph \\ 
$\tau$      & LCC degradation constraint & $k_G$       & lower bound on the num. of                      \\ 
$k$         & \# of intervention edges &             & needed edges in EORE       \\ 
$n_v$       & \# of edges among $v$'s neighbors  &$l_j$       & targeted LCC to be tried in EORE            \\ \hline
% $F$       & selected intervention edges                 \\
% $\overbar{G}$       & network after intervention                    \\ 
% $u$       & node to be connected with $t$           \\ 
% $\hat{d}$   & maximum degree in $G=(V,E)$                     \\ 
% $R$       & nodes removed by CRPD baseline   \\ 
% $r$       & a removed node in $R$                             \\ 
% $w_b$       & weights of betwenness in MISS  \\ 
% $w_c$       & weights of closeness in MISS  \\
% $w_d$       & weights of degree in MISS  \\
% $f(\cdot)$  & weight func. in threshold graph    \\ 
% $k_G$       & lower bound on the num. of \\ & needed edges in EORE   \\
% $l_j$       & targeted LCC to be tried in EORE                   \\ 
% $n_v$       & num. of edges among $v$'s neighbors            \\ \hline
 \end{tabular}
\end{table}

%\section{Network Intervention with Limited Degradation for Single target (NILD-S)}
\section{Intervention for A Single Target}
\label{sec:nild}

In this section, we first reduce the \textit{Local Clustering Coefficient (LCC)} of a targeted individual (denoted as $t$) by selecting a set of people from the social network to become friends with $t$. Given a social network $G=(V,E)$ (or $G$ for short), where each node $v \in V$ denotes an individual, and each edge $(i,j) \in E$ represents the social link between individuals $i$ and $j$, the \textit{ego network} of an individual $v$ is the subgraph induced by $v$ and its neighbors $N_G(v)$.
The LCC of a node $v$ in $G$, $LCC_G(v)$, is defined as the number of edges between the nodes in $N_G(v)$ divided by the maximum number of possible edges among the nodes in $N_G(v)$,
\small
\begin{equation}
LCC_G(v) = \frac{|\{(i,j)| i,j \in N_G(v), (i,j) \in E\}|}{C(d_G(v),2)}
\label{eq:lcc}
\end{equation}
\normalsize
where $d_G(v)=|N_G(v)|$ and $C(d_G(v),2)$ is the number of combinations to choose two items from $d_G(v)$ ones. Adding social links may increase LCC of other nodes not incident to any new edge. 

However, the increment of LCC for healthy people also needs to be carefully controlled.\footnote{According to the study on 2844 junior high school students over three years~\cite{typ}, the decrement of depression rates is significantly correlated with the reduction of LCCs with Pearson Correlation Coefficient as 0.3.} 
In addition to LCC, it is also important to ensure that the degree, betweenness, and closeness are sufficiently large. Therefore, we formulate the \textit{Network Intervention with Limited Degradation - for Single target} (NILD-S) problem as follows.

%NIND-S has no limit on the number of intervention edges allowed, but this is not practical due to the unavoidable costs. To create an intervention edge, it requires a series of meetings lasting for several weeks \cite{garrett1998arise}. Moreover, intervention specialists need to participate deeply in the intervention process. Meanwhile, not allowing the LCC of any individual to increase may limit the flexibility needed to create sufficient intervention edges for the target. However, the increment of LCC for healthy people also needs to be carefully controlled.\footnote{According to a study on 2844 junior high school students over three years~\cite{typ}, the decrement of depression rates are significantly correlated with the reduction of LCCs with Pearson Correlation Coefficient as 0.3.} In addition to LCC, it is also important to ensure that the degree, betweenness, and closeness are sufficiently high. Therefore, we formulate the Network Intervention with Limited Degradation for Single target (NILD-S) problem as follows. 

\begin{definition}
Given a social network $G=(V,E)$ (or $G$ for short), the target $t$, the number $k$ of intervention edges to be added, the LCC degradation threshold $\tau$, the lower bounds on betweenness, closeness and degree $\omega_{b}$, $\omega_{c}$, and $\omega_{d}$, NILD-S minimizes the LCC of $t$ by adding a set $F$ of $k$ edges incident to $t$, such that in the new network $\overbar{G}$, 1) $LCC_{\overbar{G}}(v) - LCC_{G}(v) \leq \tau$ for any $v$, 2) $b_{\overbar{G}}(t) > \omega_{b}$, 3) $c_{\overbar{G}}(t) > \omega_{c}$, and 4) $d_{\overbar{G}}(t) > \omega_{d}$, where $b_{\overbar{G}}(t)$, $c_{\overbar{G}}(t)$ and $d_{\overbar{G}}(t)$ are the betweenness, closeness, and degree of $t$ in $\overbar{G}$.%\footnote{The notation $\overbar{G}$ denotes $\overbar{G}(V, E\cup F)$.}        
\end{definition}

NILD-S is computationally expensive. We prove that it is NP-hard and inapproximable within any ratio, i.e., there is no approximation algorithm with a finite ratio for NILD-S unless P=NP. 
%However, later we show that NILD-S is tractable for threshold graphs, which share similar graph properties (e.g., degree distribution, largest component size, edge density, and local clustering coefficient) with many well-known online social networks, e.g., Live-Journal, Flickr, and Youtube~\cite{saha2014intergroup}. 
However, later we show that NILD-S is tractable for threshold graphs, which share similar graph properties with many well-known online social networks, e.g., Live-Journal, Flickr, and Youtube~\cite{masuda2004analysis, vernitski2012astral, saha2014intergroup}.

\begin{theorem}
NILD-S is NP-hard and cannot be approximated within any ratio in polynomial time unless P=NP. 
\label{thm:nilds_np}
\end{theorem}
% \begin{proof}
% We prove this theorem in Appendix~\ref{proof:nilds_np}.
% \end{proof}
\begin{proof}
We prove the NP-hardness by the reduction from the Maximum Independent Set (MIS) problem under triangle-free graphs (i.e., a graph without any three nodes forming a triangle)~\cite{lozin2008polynomial}. Given a triangle-free graph $G_M=(V_M,E_M)$, MIS is to find the largest subset of nodes $S_M \subseteq V_M$, such that every node in $S_M$ is not adjacent to any other nodes in $S_M$. For each instance of MIS, we construct an instance $G=(V,E)$ of NILD-S as follows. For each node $v' \in V_M$ and edge $(i',j') \in E_M$, we create the corresponding node $v \in V$ and edge $(i,j) \in E$, respectively. Also, we add a node $t$ as the targeted node and set $\tau$, $\omega_{b}$, $\omega_{c}$, and $\omega_{d}$ as 0. In the following, we prove that $G_M=(V_M,E_M)$ has an independent set $S_M$ with size $k$ in MIS if and only if the LCC of $t$ in $G=(V,E)$ remains as 0 after adding $(t, v)$ for every $v' \in S_M$. We first prove the sufficient condition. If $G_M=(V_M,E_M)$ has an independent set $S_M$ of size $k$, then there is no edge between any two nodes in $S_M$. Thus, if we add an edge $(v, t)$ for each node $v' \in S_M$ in $G(V, E)$, the LCC of $t$ is still 0. We then prove the necessary condition. If there is a set $S$ of $k$ nodes such that $t$'s LCC remains as 0 after adding $(v, t)$ to $E$, then there exists no edge among $t$'s neighbors, i.e., $S$. Therefore, $S_M$ with the corresponding nodes in $S$ is an independent set.

Next, we prove that NILD-S cannot be approximated within any ratio in polynomial time unless P=NP by contradiction. Assuming that there exists a polynomial-time algorithm with solution $lcc$ to approximate NILD-S with a finite ratio $ro$ for a triangle-free $G=(V,E)$, i.e., the LCC of the optimal solution is at least $lcc/ro$. If $lcc=0$, there is an independent set with size $k$. If $lcc>0$, the LCC of the optimal solution is at least $lcc/ro > 0$, and there is no $k$-node independent set. Thus, the approximation algorithm for NILD-S can solve MIS in polynomial time by examining $lcc$, contradicting that MIS is NP-hard \cite{lozin2008polynomial}. 
\end{proof}

%\subsection{Candidate Re-selection with Preserved Dependency}
%\vspace{-3pt}
\subsection{The CRPD Algorithm}
\label{subsec:crpd}

For NILD-S, a simple approach is to iteratively choose a node $u$, add $(t, u)$ into $F$, and eliminate (i.e., does not regard it as a candidate in the future) every neighbor $r$ of $u$ if adding both $(t, u)$ and $(t, r)$ into $F$ would increase the LCC of any node for more than $\tau$ (called the LCC degradation constraint).
%would violate the LCC degradation constraint $\tau$. 
However, the above approach does not carefully examine the structure among the neighbors of $t$. The selection of $u$ is crucial, because it may become difficult to choose its neighbors for connection to $t$ later due to the LCC degradation constraint. Therefore, a simple baseline is to extract the $u$ with the smallest degree and add $(t,u)$ to $F$, because such $u$ tends to result in the least number of neighbors removed from the pool of candidates for connecting to $t$. It removes those neighbors $r$ of $u$ if adding $(t, u)$ increases LCC of any individual by more than $\tau$. However, Example~\ref{example:CPRD_motivation} indicates that a good candidate $r$ may be improperly removed due to a small LCC.

\begin{figure}[t]
\begin{minipage}[t]{0.53\linewidth}
\centering
\includegraphics[width =0.95\linewidth]{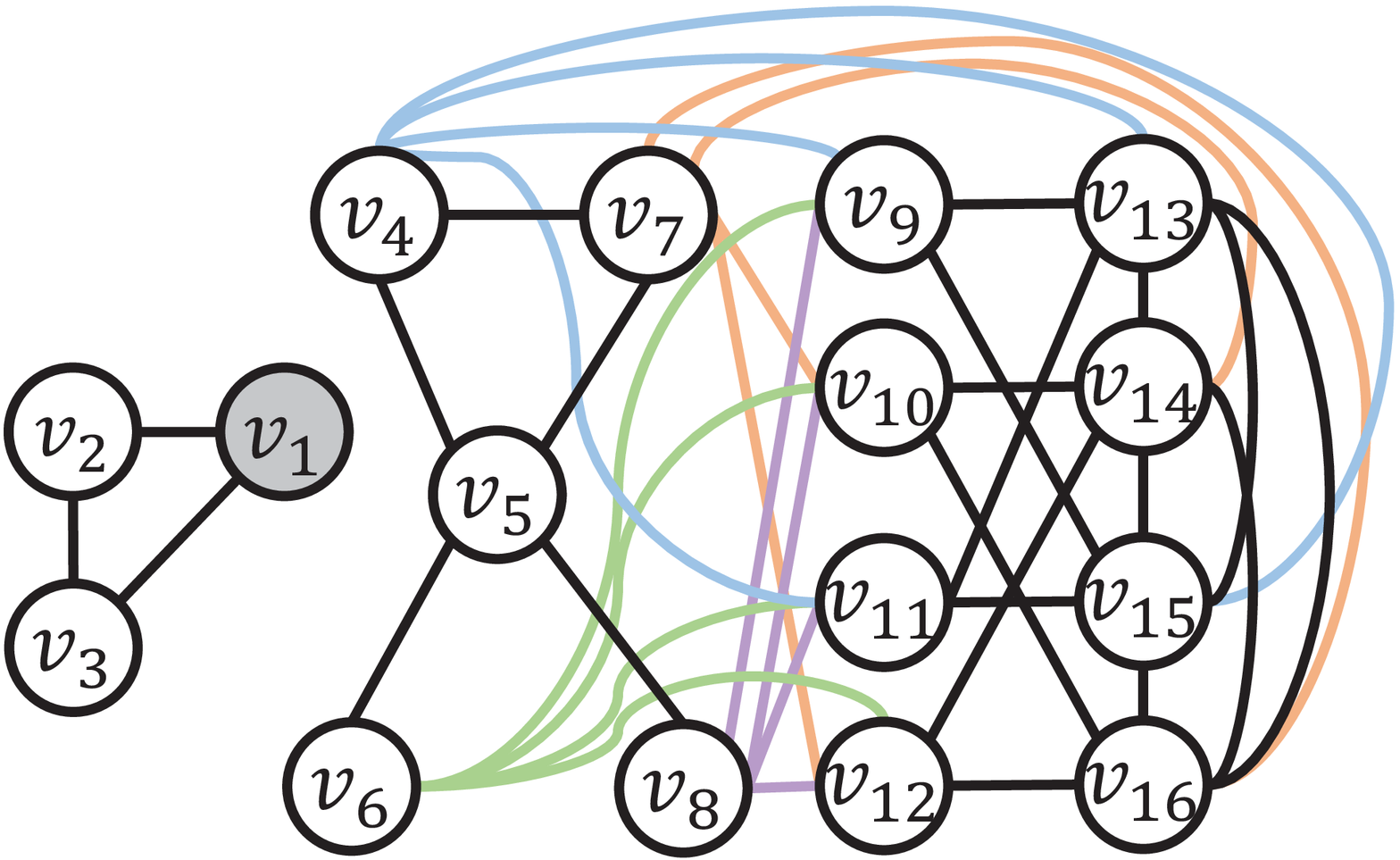}
\caption{A motivating example of CRPD}\label{fig:crpd_example}
\end{minipage}
\hspace{2pt}
\begin{minipage}[t]{0.43\linewidth}
\centering
\includegraphics[width=0.95\linewidth]{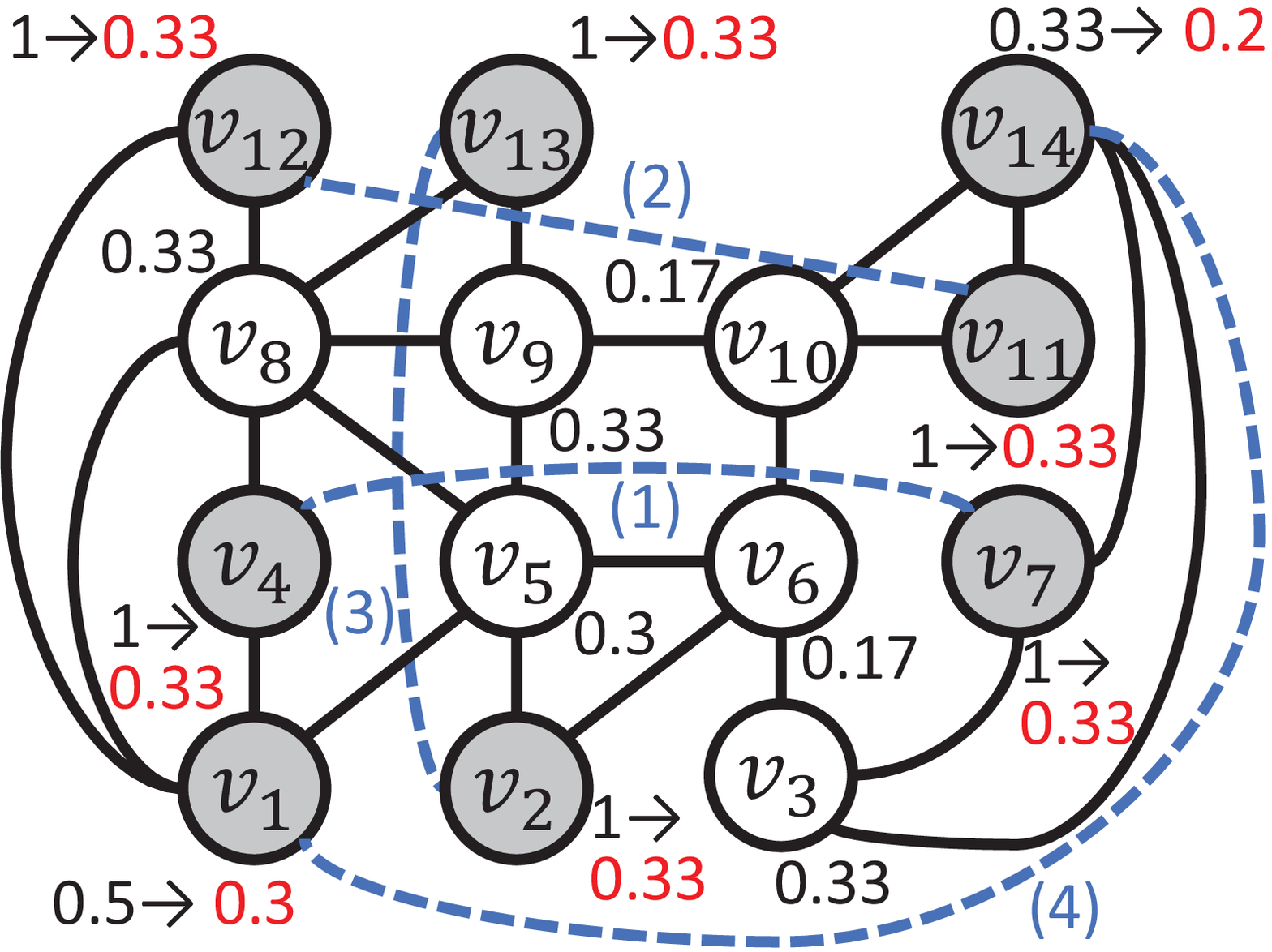}
\caption{A motivating example of OISA} \label{fig:running_example}
\end{minipage}
\end{figure}

\begin{example}
Figure~\ref{fig:crpd_example} shows an example of NILD-S with 16 nodes, where $t=v_{1}$, $k=3$, $\tau=0.05$, $\omega_{b}=0.5$, $\omega_{c}=0.5$, and $\omega_{d}=4$. Note that all edges in Figure~\ref{fig:crpd_example} are edges in $E$ regardless of their colors.  
The baseline first selects $v_5$ (i.e., adding $(v_{1}, v_5)$ to $F$) since it has the smallest degree among all nodes not connected to $v_{1}$. Then, for the neighbor $v_4$ of $v_5$, adding $(v_{1}, v_4)$ to $F =\{(v_{1}, v_5)\}$ does not increase LCC of any node to more than $\tau$. Thus, $v_4$ is still a valid candidate.\footnote{Adding $(v_{1}, v_4)$ to $F$ only changes the LCCs of $v_4$ and $v_5$, where $LCC_{\overbar{G}}(v_4) - LCC_G(v_4) = 0.33 - 0.4 < 0.05$ and $LCC_{\overbar{G}}(v_5) - LCC_G(v_5) = 0.2 - 0.17 < 0.05$.} However, after choosing $v_5$, the baseline excludes $v_6$ and $v_8$ from candidates as their respective LCCs will be increased by 0.06. Afterward, the baseline selects $v_{13}$ and $v_{14}$ to reduce the LCC of $v_1$ from 1 to 0.2. Nevertheless, a better approach is to select $v_6$, $v_8$ and $v_4$, as the LCC of $v_{1}$ can effectively diminish to 0.1.  
\label{example:CPRD_motivation}
\end{example}

\begin{algorithm}[t]
\caption{The CRPD algorithm}
\label{alg:crpd}
\begin{algorithmic}[1]
\REQUIRE $G$, $t$, $k$, $\tau$, $\omega_b$, $\omega_c$, $\omega_d$
\ENSURE A set $F$ of $k$ edges incident to $t$ to be added

\STATE{$(F, R) \leftarrow Baseline(t, k, \emptyset)$   \  //baseline}
%\STATE{$(F, R) \leftarrow Baseline(t, k, F)$} 
\STATE{Retrieve the minimum degree $r_m$ from $R$}
\STATE{$(F', R') \leftarrow Baseline(t, k, \{(t, r_m)\})$   \  //RNR}
%\STATE{$(F', R') \leftarrow Baseline(t, k, F')$ //RNR} 
\STATE{Replace $F$ with $F'$ if $F'$ is better}
\RETURN{$F$}
%\end{algorithmic}
%\vspace{-3pt}
%\end{algorithm}

%\begin{algorithm}[t]
%\caption{Baseline($t$, $k$, $F$)}
%\label{alg:crpd}
%\begin{algorithmic}[1]
\STATE{\textbf{function} Baseline($t$, $k$, $F$) //function called by CRPD}
%\bindent
\begin{ALC@g}
	%\STATE {$C \leftarrow \{i | \text{  add $(t,i)$ satisfies the $\tau$ constraint}\}$}
	\STATE {Let $C$ include all $i$ s.t. add $(t,i)$ satisfies the $\tau$ constraint}
	\STATE $R\leftarrow \emptyset$
	\WHILE {$|F| < k$}
		\STATE {Choose $u$ from $C$ according to $d_G(u)$ and Eq.~\ref{eq:miss}}
		\IF{$\exists i \in C$ s.t. add $(t,i)$ violates the $\tau$ constraint}
		    \STATE {$C\leftarrow C\backslash\{i\}$; $R\leftarrow R\cup\{i\}$ }
		%\STATE {Move $i$ from $C$ to $R$ if add $(t,i)$ violates $\tau$ constraint}
		\ENDIF
	\ENDWHILE	
    \RETURN {($F$, $R$)}
%\eindent
\end{ALC@g}
\end{algorithmic}
\vspace{-3pt}
\end{algorithm}

Motivated by Example~\ref{example:CPRD_motivation}, we propose the \textit{Candidate Re-selection with Preserved Dependency} (CRPD) algorithm for NILD-S. CRPD includes two components: 1) \textit{Removed Node Re-selection Strategy},
%2) \textit{Node Dependency Index} which accelerates the speed to update LCC and to find out if the increment exceeds $\tau$, 
and 2) \textit{Multi-measurement Integration Selection Strategy}, as follows. A pseudocode of CRPD is shown in Algorithm~\ref{alg:crpd}.

\noindent \textbf{Removed Node Re-selection (RNR) Strategy.} 
To avoid missing good candidates when processing each node $u$ in the above baseline, CRPD first extracts $R$ during the process, where each $r$ in $R$ is a neighbor of $u$ removed by the baseline due to the LCC degradation constraint, i.e., including both $(t, u)$ and $(t, r)$ in $F$ increases any node's LCC for more than $\tau$. CRPD improves the above baseline by conducting a deeper exploration that tries to replace $(t,u)$ by $(t,r_m)$, where $r_m$ is the node with the minimum degree in $R$, if selecting $(t,r_m)$ instead removes fewer neighbors and obtains a better solution later.  
In Example~\ref{example:CPRD_motivation}, adding $(v_1, v_5)$ to $F$ removes its neighbors $v_6$ and $v_8$ because including $(v_1, v_6)$ or $(v_1, v_8)$ to $F=\{(v_1, v_5)\}$ increases the LCCs of $v_6$ or $v_8$ by more than $\tau$, respectively. In contrast, adding $(v_1, v_6)$ (instead of $(v_1, v_5)$) to $F$ only removes $v_5$ and obtains a better solution $\{(v_1, v_6), (v_1, v_{8}), (v_1, v_{4})\}$, compared with the solution $\{(v_1, v_5), (v_1, v_{13}), (v_1, v_{14})\}$ of the baseline. With the above deeper inspection, later we prove that CRPD can find the optimal solution of NILD-S in threshold graphs.

\noindent \textbf{Multi-measurement Integration Selection Strategy (MISS).} 
To address $\omega_{d}$, the degree of $t$ can be examined according to $k+d_G(t)$ when more edges are included. However, when the betweenness and closeness of $t$ are smaller than $\omega_{b}$ and $\omega_{c}$, respectively, CRPD selects $u$ as follows to improve the betweenness and closeness of $t$, when multiple nodes available to be chosen as $u$.
%\small
\begin{equation}
\medmuskip=1mu
\thinmuskip=1mu
\thickmuskip=1mu
    u = \argmax_{u \notin N_G(t)}{w_b  \cdot b_G(u) + w_c \cdot c_G(u) + (1 - w_b - w_c) \cdot d_{G}(u)},    
\label{eq:miss}
\end{equation}
%\normalsize
where $b_G(u)$ and $c_G(u)$ denote the betweenness and closeness of $u$, and $w_b$ and $w_c$ are the weights of betweenness and closeness in selecting $u$. $w_b$ (or $w_c$) is
derived by computing the difference of $t$'s betweenness and $\omega_{b}$ (or $t$'s closeness and $\omega_{c}$) when $t$'s betweenness (or closeness) is smaller than $\omega_{b}$ (or $\omega_{c}$). Otherwise, $w_b$ (or $w_c$) is set as 0.
%The complexity of CRPD is analyzed in Appendix~\ref{app:crpd_com}.

\subsection{Solution Quality in Threshold Graph}
\label{subsec:tisw_optimality}

%{\color{blue}
We prove that CRPD obtains the optimal solution for \textit{threshold graphs}, which are similar to many well-known online social networks in terms of important properties like the scale-free degree distribution, short diameter, and clustering coefficient~\cite{masuda2004analysis, vernitski2012astral, saha2014intergroup}.
%}

\begin{definition}
A graph $G=(V,E)$ is a \textit{threshold graph} if there exists a weight function $f(\cdot):V \rightarrow R$ and a threshold $t_G$, such that for any two nodes $i,j \in V$, $f(i) + f(j) > t_G$ if and only if $(i,j) \in E$.
\label{def:threshold_graph}
\end{definition}
\vspace{-\abovedisplayskip}

\begin{theorem}
The CRPD algorithm can find an optimal solution of NILD-S when $G=(V,E)$ is a threshold graph, and the running time of CRPD is $O(|E|\hat{d} + k\hat{d}|V|)$, where $\hat{d}$ is the maximum degree in $G=(V,E)$. 
%(See proof in~\cite{online_version}.)
%(See proof in Appendix~\ref{proof:crpd_optimal}).
\label{thm:crpd_optimal}
\end{theorem}
With the property of threshold graphs, we prove that CRPD always finds a feasible solution, and the obtained feasible solution is one of the optimal ones in Appendix~\ref{proof:crpd_optimal}.

%\section{Network Intervention with Limited Degradation for Multiple targets (NILD-M)}
\section{Intervention for Multiple Targets}
\label{sec:problem}

We  formulate the \textit{Network Intervention with Limited Degradation - for Multiple targets} (NILD-M) problem and show the NP-hardness.%*****

% \begin{definition}
% Given a social network $G=(V,E)$ (or $G$ for short), the number of $k$ intervention edges to be added, the LCC degradation threshold $\tau$, the lower bounds on betweenness, closeness and degree $\omega_{b}$, $\omega_{c}$, and $\omega_{d}$, and a set of targeted individuals $T$, the NILD-M problem minimizes the maximal LCC among all nodes in $T$, i.e., $\max_{t \in T} LCC_{\overbar{G}}(t)$, by selecting a set of $k$ intervention edges $F$, such that 1) the LCC increment of every node $v \in V$ does not exceed $\tau$, 2) the betweenness of every $t \in T$ is greater than $\omega_{b}$, 3) the closeness of every $t \in T$ is greater than $\omega_{c}$, and 4) the degree of every $t \in T$ is greater than $\omega_{d}$.
% \end{definition}
% \vspace{-\abovedisplayskip}

\begin{definition}
Given a social network $G=(V,E)$ (or $G$ for short), the number of $k$ intervention edges to be added, the LCC degradation threshold $\tau$, the lower bounds on betweenness, closeness and degree $\omega_{b}$, $\omega_{c}$, and $\omega_{d}$, and a set of targeted individuals $T$, the NILD-M problem minimizes the maximal LCC among all nodes in $T$, i.e., $\max_{t \in T} LCC_{\overbar{G}}(t)$, by selecting a set of $k$ intervention edges $F$, such that in the new network $\overbar{G}$, 1) $LCC_{\overbar{G}}(v) - LCC_{G}(v) \leq \tau$ for any $v$, 2) $b_{\overbar{G}}(t) > \omega_{b}$ for any $t \in T$, 3) $c_{\overbar{G}}(t) > \omega_{c}$ for any $t \in T$, and 4) $d_{\overbar{G}}(t) > \omega_{d}$ for any $t \in T$, where $b_{\overbar{G}}(t)$, $c_{\overbar{G}}(t)$ and $d_{\overbar{G}}(t)$ are the betweenness, closeness, and degree of $t$ in $\overbar{G}$.
\end{definition}
\vspace{-\abovedisplayskip}

\begin{corollary}
NILD-M is NP-hard and cannot be approximated within any ratio in polynomial time unless P=NP. %(See proof in~\cite{online_version}.)
\label{col:nildm_np}
\end{corollary}
%\vspace{-\abovedisplayskip}
%\begin{proof}
Corollary~\ref{col:nildm_np} follows since NILD-S, the special case of NILD-M when $|T|=1$, is NP-hard and cannot be approximated within any ratio in polynomial time unless P=NP in Theorem~\ref{thm:nilds_np}.
%\end{proof}

%\section{Objective-aware Intervention Edge Selection and Adjustment}
\section{The OISA Algorithm}
\label{sec:algo_oisa}

A na\"ive approach for NILD-M is to exhaustively search every $k$-edge set spanning the nodes in $T$. However, the approach is not scalable as shown in Section \ref{sec:experiment}. In the following, we first present two baseline heuristics, \textit{Budget Utility Maximization (BUM)} and \textit{Surrounding Impact Minimization (SIM)} for NILD-M.  

\noindent \textbf{Budget Utility Maximization (BUM).} To intervene the maximal number of individuals within $T$, BUM repeatedly selects the node $u$ having the largest LCC without contradicting any constraints and connects $u$ to the node $m$ with the largest LCC in $T$ until $k$ edges are selected. 

\noindent \textbf{Surrounding Impact Minimization (SIM).} Without considering the proximity between $m$ and $u$, adding $(m, u)$ sometimes increases the LCCs of their common neighbors. To avoid the above situation, SIM chooses the $u$ with the maximum number of hops from $m$ in $T$ because adding $(m,u)$ is less inclined to change the LCCs of other neighbor nodes.

In summary, BUM carefully evaluates the LCC of $u$ but ignores the structural properties, whereas SIM focuses on the distance between $m$ and $u$ but overlooks the LCC. Be noted that BUM and SIM are also equipped with MISS to ensure they obtain feasible solutions.
Example~\ref{example:optimility} indicates that their solutions are far from the optimal solution.

\begin{example}
%\sloppy
Figure~\ref{fig:running_example} presents an example with $k = 4$ and $\tau = 0.15$. $G=(V,E)$ includes 14 nodes and 23 black solid edges, where each node $v$ is labeled aside by its $LCC_{G}(v)$ in dark. The targeted $T$ includes 8 grey nodes with the largest LCC in $V$, i.e., $T=\{v_2, v_4, v_7, $ \\ $v_{11},v_{12}, v_{13}, v_1, v_{14}\}$. $\omega_b$, $\omega_c$, $\omega_d$ are set as 0.02, 0.44, and 3, respectively. 
The weights $w_b$, $w_c$, and $(1-w_b-w_c)$ of MISS are set as 0.33, 0.33, and 0.34. 
For this instance, the maximal LCCs obtained by BUM and SIM are both 1. In contrast, the maximal LCC in the optimal solution is 0.66 acquired by adding the dotted blue edges into $F$. For each node $v$ whose $LCC_{\overbar{G}}(v)$ is not equal to $LCC_{G}(v)$, its $LCC_{\overbar{G}}(v)$ are labeled aside in red. The optimal solution effectively lowers the maximal LCC by 34\% from BUM and SIM without increasing any node's LCC.
%, i.e., 
%$\max_{t \in T} LCC_{\overbar{G}}(t)=0.33$.
%with $(v_4, v_{7})$, $(v_{11}, v_{12})$, $(v_2, v_{13})$, and $(v_{1}, v_{14})$.
\label{example:optimility}
\end{example}

Motivated by the strengths and pitfalls of BUM and SIM, we propose \textit{Objective-aware Intervention Edge Selection and Adjustment (OISA)} to jointly consider the LCCs and the network structure with three ideas: 1) \textit{Expected Objective Reaching Exploration (EORE)}, and 2) \textit{Poor Optionality Node First (PONF)}, and 3) \textit{Acceleration of LCC Calculation (ALC)}. 
EORE finds the minimum number of intervention edges required to achieve any targeted LCC for each candidate solution and sees if it can meet the budget constraint $k$. Given a targeted LCC, PONF carefully adds intervention edges to avoid seriously increasing LCCs for some individuals. A pseudocode of OISA is shown in Algorithm~\ref{alg:oisa}.

\subsection{Expected Objective Reaching Exploration}

EORE carefully examines the correlation between the LCC reduction and the number of intervention edges. Lemma~\ref{lemma:edge_lowerbound} first derives the minimum number of required intervention edges $k_G$ for $G$ to achieve any targeted LCC. To meet the degree constraint $\omega_{d}$, the minimum number of edges for each node $t \in T$ is at least $\omega_{d} - d_G(t)$.

% \begin{figure}[t]
% 	\centering
% 	\subfigure[Original $G=(V,E)$] {\includegraphics[width =1.3 in]{figures/running_example_1}}
% 	%\subfigure[$\sigma = 0.66$] {\includegraphics[width =1.5  in]{figures/running_example_5}}
% 	%\subfigure[Truncated $E^1_S$] {\includegraphics[width =1.5 in]{figures/running_example_5}}
% 	\subfigure[BUM] {\includegraphics[width =1.3 in]{figures/running_example_6}}
% 	\subfigure[SIM] {\includegraphics[width =1.3 in]{figures/running_example_7}}
% 	\subfigure[Optimal/OISA] {\includegraphics[width =1.3 in]{figures/running_example_8}}
% 	\caption{A running example with $k = 4$ and $\tau = 0.15$}
% 	\label{fig:running_example}
% \end{figure}

\begin{lemma}
Given a node $t$ of degree $d_{G}(t)$, to reduce \\ $LCC_{\overbar{G}}(t)$ to any targeted LCC $l$, the minimum number of intervention edges $k_{t}$ is the smallest number satisfying:
\small
\begin{equation}
\frac{LCC_G(t)\times d_{G}(t)(d_{G}(t)-1)}{(d_{G}(t) + k_{t})(d_{G}(t)+k_{t} -1)} \leq l.
\label{eq:edge_lowerbound}
\end{equation}
\normalsize
Also, the minimum number of edges $k_G$ for $G$ is \newline $0.5\cdot \sum_{t\in T} \max\{k_{t}, \omega_{d} - d_G(t)\} \leq k_G$ %(See proof in~\cite{online_version}.)
(proof in Appendix~\ref{proof:edge_lowerbound}).
\label{lemma:edge_lowerbound}
\end{lemma}
% \begin{proof}
% We prove this lemma in Appendix~\ref{proof:edge_lowerbound}.
% \end{proof}
% \begin{proof}
% We prove the lemma by contradiction. Assume there exists a $k'_{t} < k_{t}$ such that the LCC of the node $t$ can be reduced to $l$ with only $k'_{t}$ edges. First, LCC of $t$ becomes $\frac{n_{t} + n'_{t} }{C(d_{G}(t) + k'_{t},2)}$ after $k'_{t}$ edges are added and connected to $t$, where $n_{t}$ is the number of edges between $t$'s original neighbors before adding the new edges (i.e., $n_{t} = LCC_G(t)\times C(d_{G}(t), 2)$), and $n'_v$ is the number of edges between the original neighbors and the new neighbors or between any two new neighbors. Therefore, it leads to a contradiction because
% \small
% \begin{equation}
% \begin{aligned}
% \frac{n_{t} + n'_{t} }{C(d_{G}(t) + k'_{t},2)} & \geq \frac{n_{t} }{C(d_{G}(t) + k'_{t},2)} \\
% &= \frac{LCC_G(t)\times  d_{G}(t)(d_{G}(t)-1)}{(d_{G}(t) + k'_{t})(d_{G}(t)+k'_{t} -1)} 
% > l,
% \end{aligned}
% \end{equation}
% \normalsize
% when $k'_{t} < k_{t}$. Note that $k_{t}=0$ if the LCC of $t$ is already smaller than $l$. Therefore, $0.5\times \sum_{t \in T} \max\{k_{t}, \omega_{d} - d_G(t)\}$ is a lower bound on $k_G$ for $l$ while satisfying the degree constraint $\omega_{d}$, since an edge connects two nodes.
% \end{proof}

Equipped with Lemma~\ref{lemma:edge_lowerbound}, a simple approach is to examine every possible LCC, i.e., $\frac{1}{C(\hat{d},2)}, \frac{2}{C(\hat{d},2)}, ..., 1$, where $\hat{d}$ is the maximum degree among all nodes in $G$. However, since adding $k$ edges can make intervention for at most $2k$ nodes, it is not necessary to scan all possible LCCs, and EORE thereby only examines a small number of targeted LCCs $l_j$, where $l_j = \frac{j}{C(\hat{d}_{2k},2)}$ with $j = 1, 2, ...$, until $l_j$ exceeds the maximal LCC before intervention, and $\hat{d}_{2k} \leq \hat{d}$ is the maximum degree among all top-$2k$ nodes with the largest LCCs in $T$. OISA skips an $l_j$ if $k_G > k$ according to the above lemma. 

For any targeted LCC $l_j$, every node $t\in T$ requires at least $\max\{k_{t}, \omega_d - d_G(t)\}$ intervention edges to achieve the targeted LCC $l_j$ and the degree constraint. Thus, OISA stops the edge selection process if there exists a node $t$ not able to achieve the above goals when $k - |F| + x_{t} < \max\{k_{t}, \omega_d - d_G(t)\}$, where $x_{t}$ is the number of intervention edges incident with $t$ in $F$.

% ***** To be added back in online version
\begin{example}
\label{eg:eore}
For the example in Figure~\ref{fig:running_example}, the node with the maximum degree among the top-$2k$ nodes with the largest LCCs is $v_{1}$, where its degree is 4. The targeted LCCs to be examined (i.e., $l_j$) are %0.27, 0.33, 0.4, 0.47, ..., 1.
0.17, 0.33, 0.5, 0.67, 0.83, 1.
For the targeted LCC $l_1 = 0.33$, $k_G$ is $7/2 = 3.5$ since $k_{t}=0$ for $v_{14}$, and $k_{t} = 1$ for $v_1$, $v_2$, $v_4$, $v_7$, $v_{11}$, $v_{12}$, and $v_{13}$. However, for the targeted LCC $l_2 = 0.17$, $k_G$ is $(2 \times 7 + 1)/2 = 8.5$ since $k_{t} =1$ for $v_{14}$, and $k_{t} = 2$ for $v_1$, $v_2$, $v_4$, $v_7$, $v_{11}$, $v_{12}$, and $v_{13}$. Thus, it is impossible for the maximal LCC to reach 0.17. 
\end{example}

\subsection{Poor Optionality Node First}
\label{subsubsec:ponf}

Recall that BUM ignores the proximity between the two terminal nodes of an intervention edge, and SIM does not examine the LCCs of both terminals and their nearby nodes. Most importantly, both strategies do not ensure the LCC degradation constraint. To address this critical issue, for each targeted LCC, we first propose the notion of \textit{optionality} to identify qualified candidate intervention edges that do not increase the LCC of any individual to more than $\tau$.

\begin{definition}
\textbf{Optionality.} For a target $t$, the optionality of $t$ denotes the number of nodes in the option set $U_{t} \subseteq T$ such that for every $u_t \in U_{t}$, either 1) the hop number from $u_t$ to $t$ is no smaller than 3, or 2) $u_t$ is two-hop away from $t$ and adding an edge $(t, u_t)$ does not increase the LCC of any common neighbor by more than $\tau$.
\label{def:optionality}
\end{definition}

%and $LCC_G(c) < \tau$ for every common neighbor $c \in N_G(u) \cap N_G(v)$.%

For the first case, adding an intervention edge $(t, u_t)$ does not increase the LCC of any node. For the second case, the LCC degradation constraint can be ensured as long as the LCCs of common neighbors are sufficiently small. Equipped with the optionality, each iteration of PONF first extracts the node $m$ with the largest LCC in $T$. If there are multiple candidates for $m$ with the same LCC, (e.g., $v_2$ and $v_4$ in Figure~\ref{fig:running_example}), 
PONF selects the one with the smallest optionality as $m$ so that others with larger optionalities can be employed later. In contrast, if $m$ was not selected by now, its optionality tends to decrease later when the network becomes denser and may reach 0, so that the LCC of $m$ can no longer be reduced without increasing the LCCs of others, boosting the risk to violate the LCC degradation constraint. The node $u$ of the intervention edge $(m,u)$ is the one with the largest LCC in the option set $U_{m}$ to reduce the LCCs of both $m$ and $u$. PONF also exploits MISS in Section~\ref{subsec:crpd} to choose $u$ according to the differences between $t$'s betweenness and $\omega_{b}$, and between $t$'s closeness and $\omega_{c}$. %***\textbf{}

% ***** To be added back in online version
\begin{example}
For the example in Figure~\ref{fig:running_example} with the targeted LCC as 0.33, one intervention edge is selected for $v_2$, $v_4$, $v_7$, $v_{11}$, $v_{12}$ and $v_{13}$. In the first iteration, the optionalities of nodes $v_{2}$, $v_{4}$, $v_{7}$, $v_{11}$, $v_{12}$ and $v_{13}$ are 7, 5, 5, 5, 6, and 7, respectively. The option set of $v_{4}$ is $\{v_{2}, v_7, v_{11}, v_{13}, v_{14}\}$. Thus, PONF chooses $v_{4}$ as $m$ and $v_{7}$ as $u$. Note that $v_7$ is the node with the largest difference to reach $\omega_{b}$, $\omega_{c}$, and $\omega_{d}$ according to Equation~\ref{eq:miss}. In the second iteration, the optionality of $v_{2}$, $v_{11}$, $v_{12}$, and $v_{13}$ are 7, 5, 6, 7, respectively. Thus, PONF selects $v_{11}$ as $m$ and $v_{12}$ as $u$. It repeats the above process and chooses $(v_2, v_{13})$ and $(v_1, v_{14})$ afterward. Figure~\ref{fig:running_example} shows the returned solution with the maximal LCC as 0.33. It is also the optimal solution in this example.
\end{example}

\begin{algorithm}[t]
\caption{The OISA algorithm}
\label{alg:oisa}
\begin{algorithmic}[1]
\REQUIRE $G$, $T$, $k$, $\tau$, $\omega_b$, $\omega_c$, $\omega_d$
\ENSURE A set $F$ of $k$ edges to be added, such that the maximal LCC among nodes in $T$ are minimized, while the LCC increment of any node does not exceed $\tau$ and the betweenness, closeness, and degree of all target nodes exceed $\omega_b$, $\omega_c$, $\omega_d$.

%\STATE {//Identify a list of targeted LCC by EORE}
\STATE{$j\leftarrow 1$}
\WHILE{$l_j = \frac{j}{C(\hat{d}_{2k},2)}$ < the maximum LCC of nodes in $G$}
    \FOR {every $t \in T$}
        \STATE{Calculate $k_{t}$ according to Lemma~\ref{lemma:edge_lowerbound}}
    \ENDFOR
    \STATE{Set $k_G$ as sum of $\max\{k_{t}, \omega_d - d_G(t)\}$ of all $t \in T$}
    \IF{$k_G < k$}
        \STATE{$F \leftarrow \emptyset$ //Enter the PONF process }
        \FOR {$i = 1...k$}
            \STATE{Choose $m$ as the node with the maximal LCC in $T$}
            \STATE{Choose $u$ according to Definition~\ref{def:optionality} and Eq.~\ref{eq:miss}}
            \STATE{Add $(m, u)$ into $F$}
            \STATE{Recompute nodes' LCC with the acceleration of ALC}
        \ENDFOR
        \STATE{Record $F$ if it reaches a smaller maximal LCC}
    \ENDIF
    \STATE{$j \leftarrow j + 1$}
\ENDWHILE
%\STATE{$l_j = \frac{j}{C(\hat{d}_{2k},2)}$ with $j = 1, 2, ...$, until $l_j$ exceeds the maximal LCC before intervention}
\RETURN{The best $F$ found}
\end{algorithmic}
%\vspace{-3pt}
\end{algorithm}

\subsection{Acceleration of LCC Calculation}
 
When an edge $(m, u)$ is added, only the LCCs of $m$ and $u$ and their common neighbors are likely to change. However, the LCC update cost of $m$ is not negligible when the number of neighbors is huge, since it needs to examine whether there is an edge from $u$ to each of $m$'s neighbors. The adjacency list of every $m$'s neighbor $v_c$ is required to be inspected even when the LCC of $v_c$ remains the same. To improve the efficiency, ALC avoids examination of every node $v$ by deriving its LCC upper bound $\overline{LCC}(v,k)$.

\begin{definition}
The LCC upper bound of $v$ after adding $k$ edges is 
\small
\begin{equation}
\overline{LCC}(v,k) = \max_{k_1+k_2 = k} \{ \frac{n_v + k_2 +k_1 \times d_G(v) + C(k_1,2)}{C(d_G(v) + k_1, 2)}\},
\end{equation}
\normalsize
where $d_G(v)$ is the degree of $v$ in $G$. $n_v$ is the number of edges between $v$'s neighbors before intervention, $n_v = LCC_G(v) \times C(d_G(v), 2)$.
\label{def:upper_bound_LCC}
\end{definition}

\vspace{-3pt}

\begin{theorem}
For an intervention edge set $F$ of size $k$, \\ $LCC_{\overbar{G}}(v) \leq \overline{LCC}(v,k)$ holds.
%(proof in Appendix~\ref{proof:upper_bound_relation}).  
\label{lemma:upper_bound_relation}
\end{theorem}
\begin{proof}
%We prove the theorem by contradiction. Assume that $LCC_{G(V, E \cup F)}(v) > \overline{LCC}(v,k)$. 
%\textbf{why the follwoing is $k'_2$, instead of $k_2$? It may lead to ambiguity.}
Let $k_1$ and $k_2$ denote the numbers of intervention edges connecting to $v$ and any two neighbors of $v$, respectively. After intervention, $k_1 + k_2 \leq k$, and $LCC_{\overbar{G}}(v) = \frac{n_v + k_2 + y}{C(d_G(v) + k_1, 2)}$, where $y$ is the number of edges between the new neighbors via the $k_1$ new edges and the original neighbors of $v$ in $G$, and $y \leq d_G(v) \times k_1 + C(k_1, 2)$. Thus, $LCC_{\overbar{G}}(v) \leq \frac{n_v + k_2 + d_G(v) \times k_1 + C(k_1, 2)}{C(d_G(v) + k_1, 2)}$. Let $k_3 = k - k_1 - k_2$.
Then, we obtain $LCC_{\overbar{G}}(v)\leq$ $\frac{n_v + (k_2 +k_3)+ k_1 \times d_G(v) + C(k_1, 2) }{C(d_G(v) + k_1, 2)}$ \\
$\leq  \max_{k_1+k_2= k} \{ \frac{n_v + k_2 +k_1 \times d_G(v) + C(k_1, 2)}{C(d_G(v) + k_1, 2)}\}$  $=\overline{LCC}(v,k)$. 
\end{proof}

According to the above theorem, ALC first derives $LCC_{G}(v)$ and $\overline{LCC}(v,k)$ as a pre-processing step of OISA before intervention. Accordingly, PONF does not update the LCC of a node $v$ if the intervention edge neither connects to $v$ nor spans $v$'s two neighbors, and $v$ is not going to be $m$ and $u$ in the next iteration since $\overline{LCC}(v,k)$ is smaller than the current maximal LCC potential to be $m$ and $u$ in the next iteration. Therefore, Theorem~\ref{lemma:upper_bound_relation} enables OISA to effectively skip the LCC updates of most nodes.

\begin{theorem}
The time complexity of OISA is $O(n_l \times k \times (n_m \times (|V| + |E|)))$, where $n_l$ is the number of targeted LCCs, and $n_m$ is the number of nodes with the largest LCC
%(See proof in~\cite{online_version}.) 
(proof in Appendix~\ref{proof:oisa_complexity}).
\label{thm:oisa_complexity}
\end{theorem}

\section{Experimental Results}
\label{sec:experiment}

We evaluate the effectiveness and efficiency of CRPD and OISA by experimentation in Sections~\ref{subsec:crpd_evaluation} and \ref{subsec:oisa_evaluation}, respectively. Also, to show the feasibility of using OISA in real-world setting, we present an empirical study in Section~\ref{sec:human_study}.\footnote{The IRB number is 10710HE072.} The study has been inspected by 11 clinical psychologists and professors in the field.\footnote{They are from California School of Professional Psychology, Taipei City Government Community Mental Health Center, National Taipei University of Nursing and Health Science, National Taiwan University etc.\label{fnlabel}}
%\footref{fnlabel}

\begin{table}[t]
\centering
\footnotesize
\caption{Dataset statistics}
\label{table:dataset}
\setlength{\tabcolsep}{1mm}{
\begin{tabular}{|lrrr|lrrr|}
\hline
dataset & $|V|$ & $|E|$ & {\footnotesize ave. LCC} & dataset &$|V|$ & $|E|$ & {\footnotesize ave. LCC} \\ \hline\hline
\textit{CPEP} \cite{gest14teacher}& 226 & 583 & 0.71&\textit{Youtube} \cite{SNAP} & 1.1M   & 3.0M   & 0.08\\ \hline
\textit{Facebook} \cite{viswanath2009evolution} & 60.3K& 1.5M & 0.22 & \textit{Amazon} \cite{SNAP} & 33.5K & 92.6K & 0.40 \\ \hline
\textit{Flickr} \cite{mislove2007measurement}  & 1.8M & 22.6M  & 0.25 &\textit{Cond-Mat} \cite{SNAP} & 23.1K & 93.5K  & 0.63  \\ \hline
\end{tabular}
}
\end{table}

%\todo[inline]{Add subsection headings to more clearly highlight how many experiments you have in total, including simulation studies and empirical study. Also be clear on why you need these.}

%We evaluate the effectiveness and efficiency of the proposed algorithms on real datasets and present an empirical study in Section~\ref{sec:human_study}.
%\footnote{The codes and data for reproducibility are on \url{https://4wluzh2irgny72cwtm38aa-on.drv.tw/245/} with proper anonymization.}
%\footnote{The code and data are on \url{https://4wluzh2irgny72cwtm38aa-on.drv.tw/245/}.}
%in Appendix~\ref{sec:human_study}. 
For simulation, since there is no prior work on lowering LCCs while ensuring their betweenness, closeness and degree, we compare the proposed CRPD and OISA with five baselines: 1) Budget Utility Maximization (\textit{BUM}): BUM iteratively adds an edge between a targeted node and the node with the largest LCC, while not violating the constraints;  2) Surrounding Impact Minimization (\textit{SIM}): SIM iteratively adds an edge from a targeted node to the node with the maximal number of hops from it, while not violating the constraints; 3) Enumeration (\textit{ENUM}): ENUM exhaustively finds the optimal solution. %Moreover, as there is no existing work on improving the LCC of targets, we adopt the following baselines, which improve one of the constraints, closeness, as baselines:
4) Edge Addition for Improving Network Centrality (\textit{EA})~\cite{papagelis2015refining}: EA iteratively adds an edge with the largest increment on closeness centrality, and 5) Target-oriented Edge Addition for Improving Network Centrality (\textit{TEA})~\cite{crescenzi2016greedily}: TEA iteratively adds an edge with the largest increment on closeness centrality for the targeted nodes. 
%{\color{blue} 
6) Greedy algorithm for Dyad scenario (\textit{GD})~\cite{wilder2018optimizing}: GD iteratively adds an edge with the largest increment on influence score, where the influence score is calculated in a similar way to PageRank. %\footnote{The edge selected by GD to be removed remains in the network because edge deletion is not considered in NILD-S and NILD-M. \textbf{reviewers may think that it is not fair for GD because we make it cripple....}}
%}
All algorithms are implemented on an HP DL580 G9 server with four Intel Xeon E7-8870v4 2.10 GHz CPUs and 1.2 TB RAM. 
Six real datasets are evaluated in the experiments. The first one, \textit{CPEP}~\cite{gest14teacher}, contains the complete social network of the students in 10 classrooms of several public elementary schools in the US. The other five large real social network datasets, collected from the Web, are \textit{Facebook}, \textit{Flickr}, \textit{Youtube}, \textit{Amazon}, and \textit{Cond-Mat}. Some statistics of datasets used in our experiments are summarized in Table~\ref{table:dataset}. The default $\tau$, $\omega_b$ and $\omega_c$ are set to $0.12$, $0.01$, and $0.1$, respectively.

\vspace{-9pt}
\subsection{Evaluation of CRPD for NILD-S}
\label{subsec:crpd_evaluation}

\begin{figure*}[t]
\centering
\subfigure[$t$'s LCC on \textit{CPEP}] {\includegraphics[width=1.35 in]{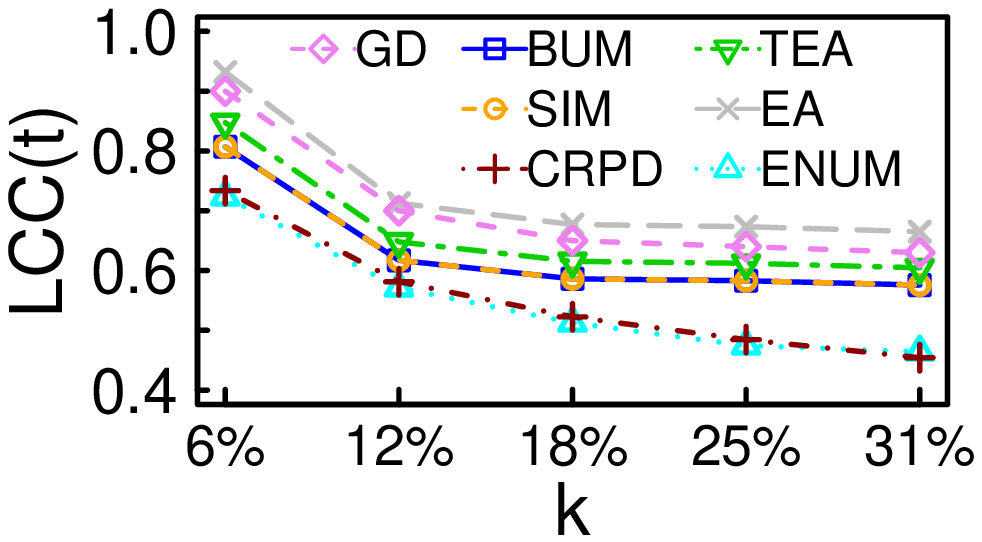}}  
%\subfigure[$t$'s LCC on \textit{Facebook}] {\includegraphics[width=1.59 in]{figures/misw_fb_obj}}
\subfigure[$t$'s LCC on \textit{Flickr}] {\includegraphics[width=1.35 in]{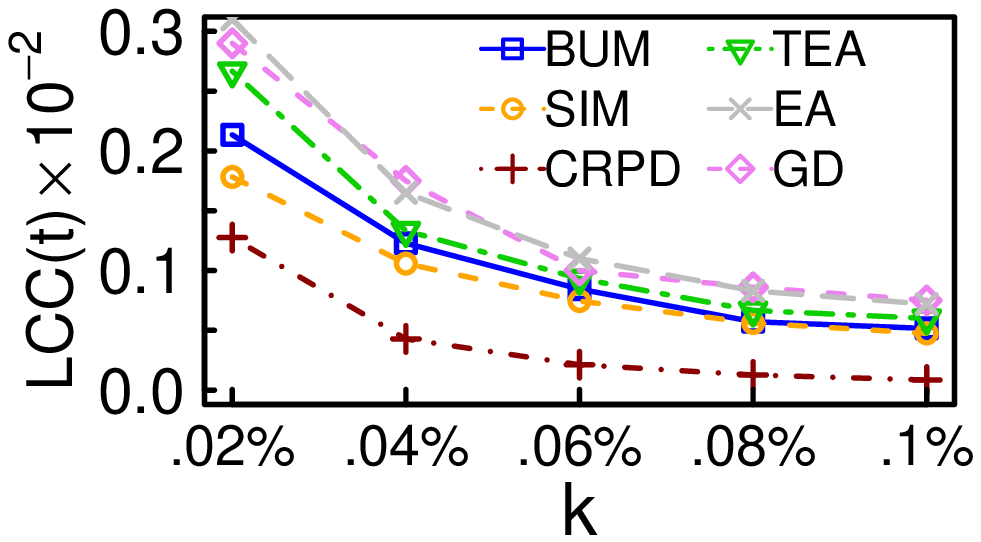}}
\subfigure[Time on \textit{Flickr}] {\includegraphics[width=1.35 in]{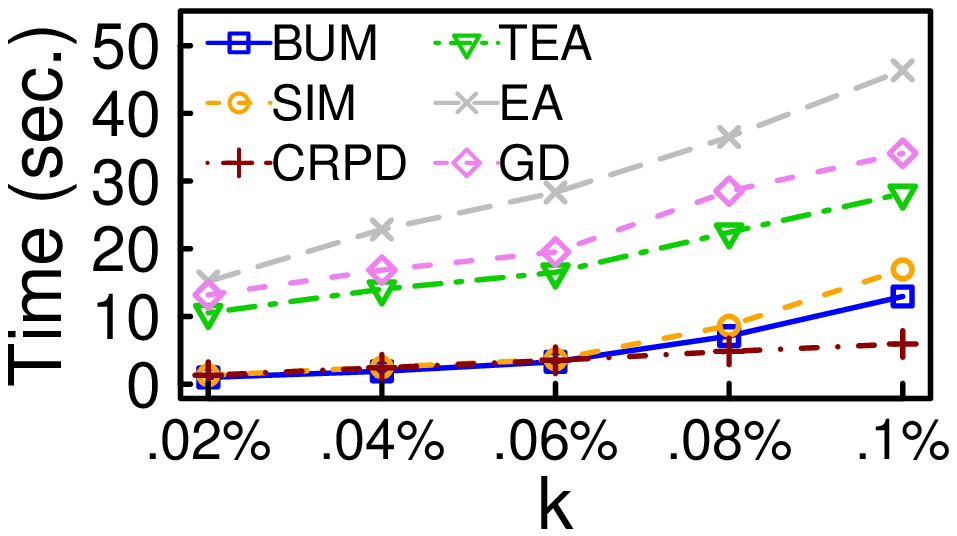}}
%\vspace{-2pt}
\subfigure[$t$'s betweenness on \textit{Flickr}] {\includegraphics[width=1.35 in]{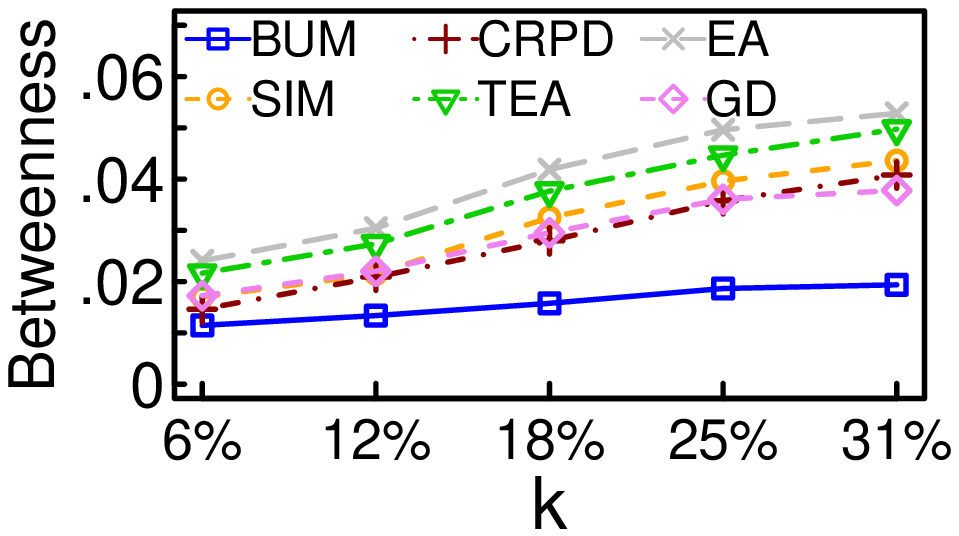}}  
\subfigure[$t$'s closeness on \textit{Flickr}] {\includegraphics[width=1.35 in]{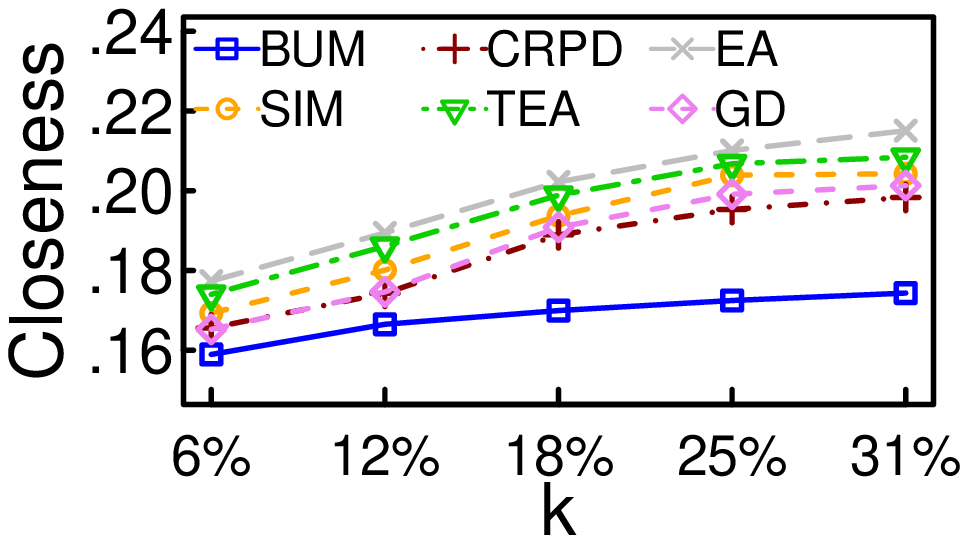}}
\subfigure[Varying $\tau$ on \textit{Flickr}]{\includegraphics[width=1.35 in]{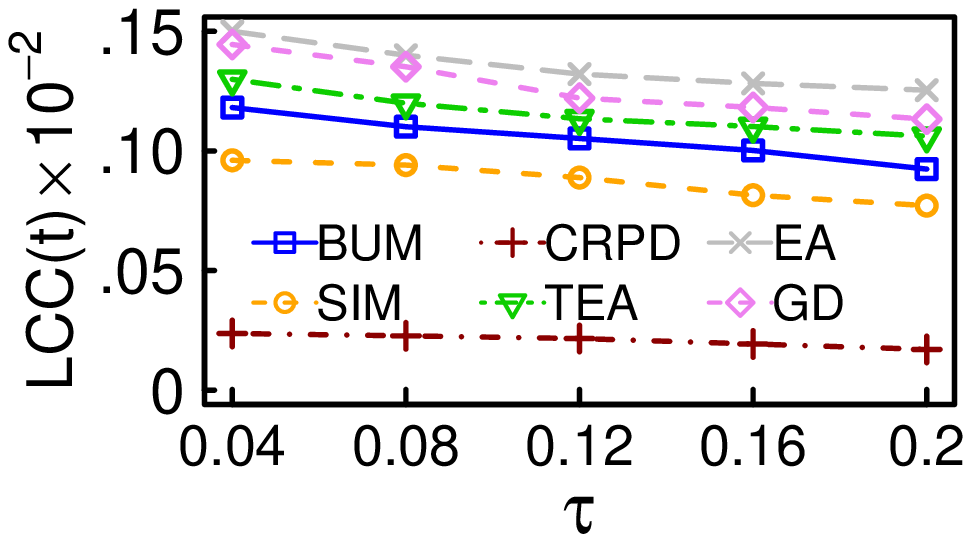}} 
\subfigure[Varying $\tau$ on \textit{Flickr}]{\includegraphics[width=1.35 in]{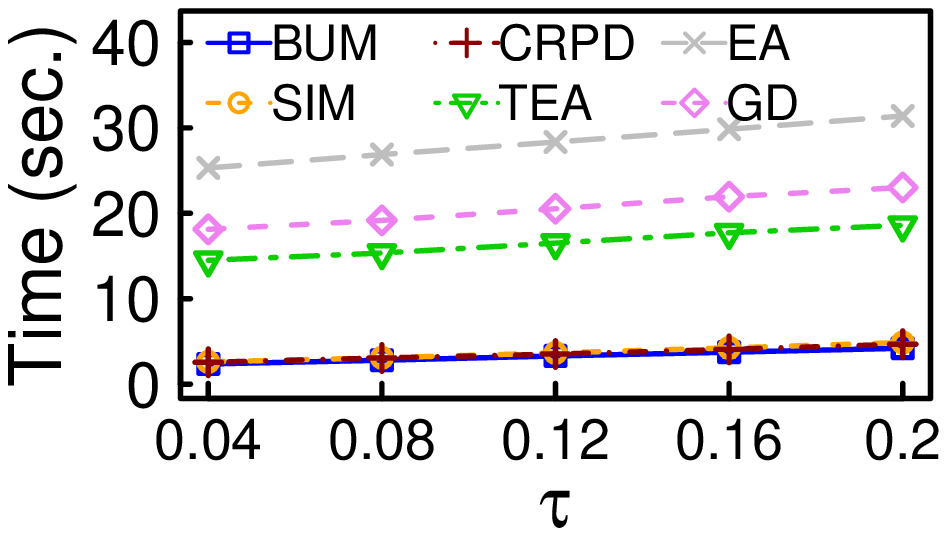}} 
\subfigure[Varying $\omega_{b}$ on \textit{Flickr}]{\includegraphics[width=1.35 in]{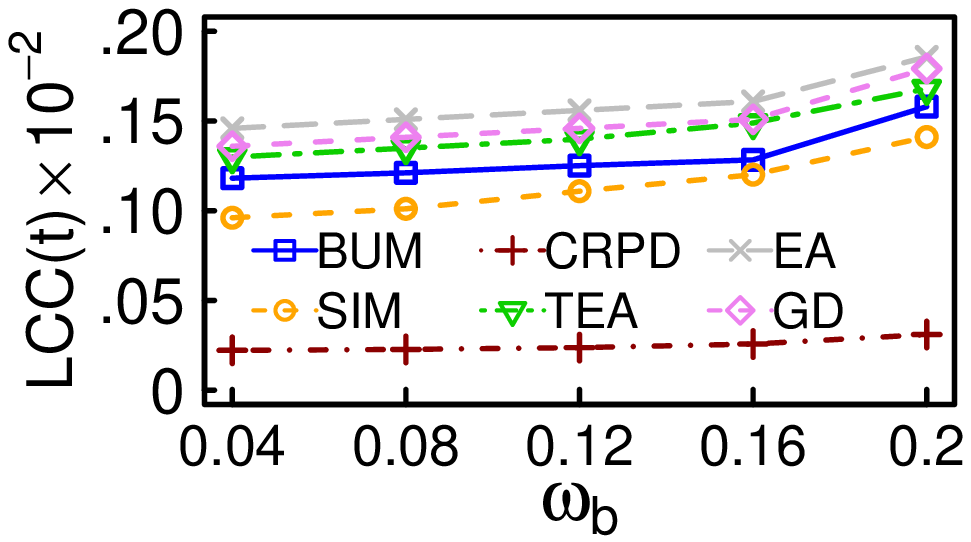}}
%\subfigure[Varying $\omega_{c}$ on \textit{Flickr}]{\includegraphics[width=1.5 in]{figures/misw_fl_obj_clo.eps}}
\subfigure[$t$'s LCC on \textit{Flickr}]{\includegraphics[width=1.35 in]{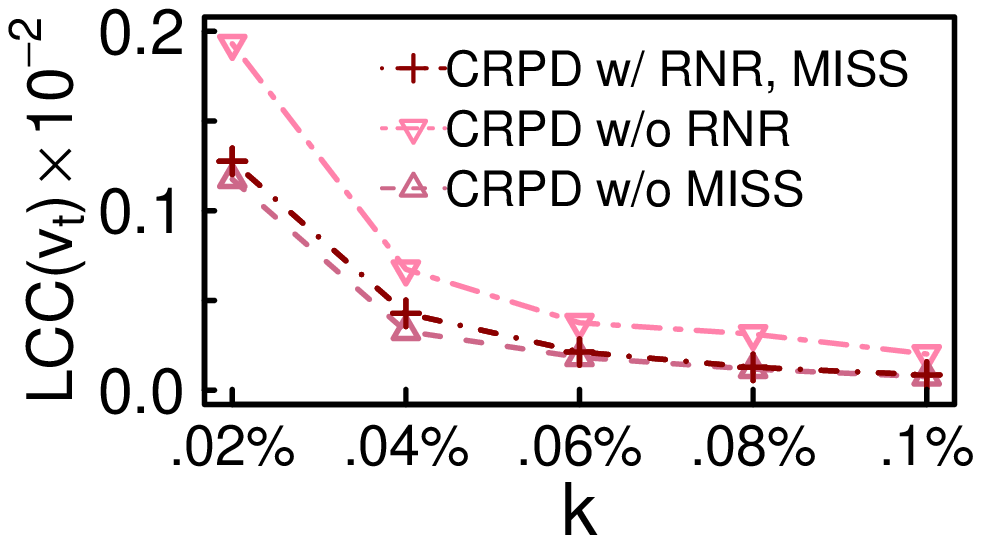}}
\subfigure[$t$'s closeness on \textit{Flickr}]{\includegraphics[width=1.35 in]{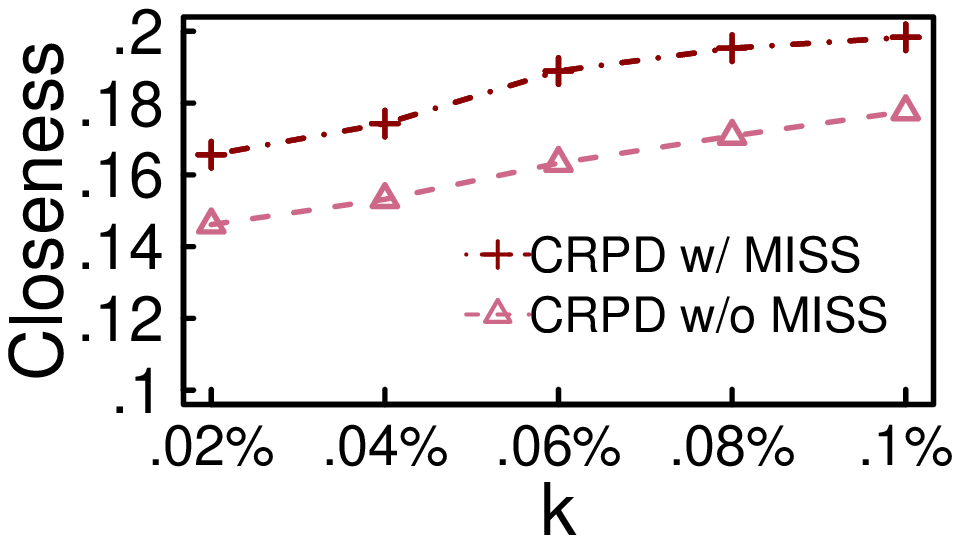}}
\caption{Sensitivity tests for NILD-S}
%\vspace{-8pt}
\label{fig:tisw_k}
\vspace{-10pt}
\end{figure*}

In the following, we first compare CRPD with baselines by varying $k$. For \textit{Facebook} and \textit{Flickr}, ENUM does not return the solutions in two days even when $k = 10$, and thus it is not shown. 
%Figures~\ref{fig:tisw_k}(a) and \ref{fig:tisw_k}(b) compare the LCC of $t$ with various $k$ on \textit{CPEP} and \textit{Flickr}, respectively, under default $\tau$, $\omega_b$ and $\omega_c$
%in Figure~\ref{fig:tisw_k}(a)-(e). 
Figures~\ref{fig:tisw_k}(a)-(e) compare the results of various $k$ on \textit{CPEP} and \textit{Flickr} under default $\tau$, $\omega_b$ and $\omega_c$.
For \textit{CPEP}, $k$ is set to 6\%, 12\%, 18\%, 25\%, and 31\% of the number of nodes (i.e., 1, 2, 3, 4, and 5 intervention edges). 
%For \textit{Facebook}, $k$ is set to 0.16\%, 0.33\%, 0.50\%, 0.66\%, and 0.83\% of nodes (i.e., 100, 200, 300, 400, and 500 intervention edges). 
For \textit{Flickr}, $k$ is set to 0.02\%, 0.04\%, 0.06\%, 0.8\%, and 0.1\% of nodes (i.e., 400, 800, 1200, 1600, and 2000 intervention edges). %Moreover, $\tau$ is set as 0.12. 
%Note that the above percentage is computed by the number of intervention edges over the total number of nodes $|V|$, because for any $t$, there are at most $|V|$ candidate edges that can be connected with $t$.
The $t$ is randomly chosen from the nodes with LCC larger than 0.8, and we report the average of 50 trials.
Figures~\ref{fig:tisw_k}(a) and \ref{fig:tisw_k}(b) show that with $k$ increasing, $t$'s LCC decreases as more edges are connected with $t$ to reduce its LCC. Also, CRPD outperforms other baselines as it carefully examines the candidates' structure to avoid increasing the number of edges among $t$'s neighbors.
%Figure~\ref{fig:tisw_k}(c) shows the running time on \textit{Flickr}. 
Figure~\ref{fig:tisw_k}(c) shows the running time of OISA is comparable with simple baselines BUM and SIM while achieving better LCC on \textit{Flickr}. 
%\footnote{The running time of \textit{CPEP} and \textit{Facebook} shares a similar trend.} 
%As shown, CRPD is much more efficient due to the help of NDI.

Figures~\ref{fig:tisw_k}(d) and \ref{fig:tisw_k}(e) show the betweenness and closeness of $t$ after adding the intervention edges to \textit{Flickr}. %, where thresholds $\omega_{b}=0.01$ and $\omega_{c}=0.1$.
EA and TEA obtain the largest betweenness and closeness, but their maximal LCCs are not effectively reduced.  
SIM also achieves large betweenness and closeness since it selects the node farthest from $t$ as $u$ and creates a shortcut between them, but the maximal LCC of SIM does not significantly decrease. 
GD achieves smaller betweenness and closeness than SIM since it generally selects a node $u$ with a large PageRank score, but its improvement on $t$'s betweenness and closeness is smaller than SIM. 
BUM induces the smallest betweenness and closeness of $t$ since it selects $u$ with the largest LCCs, and those chosen $u$ are inclined to be near each other. In contrast, CRPD achieves comparable performance with EA and TEA, showing that CRPD can effectively improve not only LCC but also other network characteristics.
% version of CPEP
%The betweenness and closeness of CRPD and ENUM both fall between those of BUM and SIM, indicating that $u$ chosen by CRPD shares a similar distance with $u$ of ENUM. 
%However, CRPD outperforms ENUM in betweenness and closeness by MISS. The above results manifest that CRPD can effectively improve not only LCC but also other important network characteristics. 
%\textbf{ENUM is not shown in the figure.....you need to check CAREFULLY about the details....}

%\subsubsection{Sensivity Tests \& Components Effectiveness}
Next, we conduct a series of sensitivity tests on $\tau$, $\omega_b$, and $\omega_c$, and show the component effectiveness. The results on \textit{CPEP} and \textit{Facebook} are similar to \textit{Flickr} and thus not shown here.

\noindent \textbf{Varying $\tau$. }Figure~\ref{fig:tisw_k}(f) compares LCC of $t$ with different $\tau$ on \textit{Flickr}, under default 
$\omega_{b}$ and $\omega_{c}$.
As shown, the LCC of $t$ decreases because a large $\tau$ allows more nodes to be candidates and thereby tends to generate a better solution. %%% start from here
%which have more connections with $t$' neighbors and would increase the connectivity between $t$'s neighbors if the nodes are connected with $t$.  
%However, the increment of CRPD with increasing $\tau$ is much less than the increment of other baselines, because CRPD, inherited from MFCR, selects the candidate that has the least change on other nodes' LCC.  
Figure~\ref{fig:tisw_k}(g) shows that the running time of CRPD and baselines on \textit{Flickr} increases as $\tau$ grows because more candidates are considered. 
%***HJ: remove the following line if NDI is not put back***
%CRPD outperforms EA and TEA because NDI effectively accelerates the LCC evaluation under different $\tau$.

\noindent \textbf{Varying $\omega_b$ and $\omega_c$. } Figure~\ref{fig:tisw_k}(h) compares the LCC of $t$ with different $\omega_{b}$ on \textit{Flickr}, where $k=0.04\%$ (800 intervention edges). The trend of $\omega_{c}$ is similar to that of $\omega_{b}$ and thus not shown here. LCC of $t$ slightly grows with increasing $\omega_{b}$ and $\omega_{c}$, because large $\omega_{b}$ and $\omega_{c}$ require CRPD and baselines to allocate more edges for the betweenness and closeness, instead of focusing on reducing LCC.
%, which decreases the amount of LCC reduction caused by selected edges. 
Moreover, LCC of $t$ obtained by TEA, EA, and CRPD increases slower than one obtained by BUM, SIM and GD, since the edges selected by TEA and EA maximize the closeness centrality while fulfilling the constraints of $\omega_{b}$ and $\omega_{c}$, whereas the edges from CRPD effectively increase the betweenness and closeness of $t$.

\noindent \textbf{Component Effectiveness. }
Figure~\ref{fig:tisw_k}(i) compares LCC of $t$ obtained by CRPD with and without RNR, and Figure~\ref{fig:tisw_k}(j) compares the closeness of CRPD with and without MISS on \textit{Flickr}. The trend of betweenness is similar to closeness and thereby not shown here. As shown, while CRPD with and without MISS shares similar LCCs, CRPD with MISS achieves a greater closeness because choosing $u$ with larger closeness/betweenness tends to increase $t$'s betweenness and closeness.
The results of CRPD on \textit{CPEP} and \textit{Facebook} share a similar trend and thus are not shown here.

\vspace{-3pt}
\subsection{Evaluation of OISA for NILD-M}
\label{subsec:oisa_evaluation}

%In the following, we evaluate the proposed OISA where the top-20\% of nodes with the maximum LCC are chosen as the default $T$.
%{\color{blue}
In the following, we evaluate the proposed OISA by randomly choosing 20\% of nodes in $V$ as $T$ from the nodes with the top-40\% maximum LCCs. The results are the average of 50 trials.
%}

% \begin{figure}[t]
% \centering
% \subfigure[Maximal LCC] {\includegraphics[width=1.35 in]{figures/cpep_obj}}  
% \subfigure[Running time]{\includegraphics[width=1.35 in]{figures/cpep_time}}  
% \subfigure[Change of LCC]{\includegraphics[width=2.8 in]{figures/cpep_lcc_distribution}} 
% \vspace{-2pt}
% \caption{Varying $k$ for NILD-M on \textit{CPEP}}
% \label{fig:cpep}
% \end{figure}

% \begin{figure*}[t]
% \centering
% \subfigure[Maximal LCC on \textit{Flickr}] {\includegraphics[width=1.35 in]{figures/fl_obj}}
% \subfigure[Time on \textit{Flickr}]{\includegraphics[width=1.35  in]{figures/fl_time}}
% \subfigure[LCC change on \textit{Flickr}]{\includegraphics[width=2.7 in]{figures/fl_lcc_distribution}} 
% \subfigure[Ave. betweenne of $T$ on \textit{Flickr}] {\includegraphics[width=1.35 in]{figures/nildm_fl_bet}}
% \subfigure[Ave. closeness of $T$ on \textit{Flickr}]{\includegraphics[width=1.35 in]{figures/nildm_fl_clo}} 
% \subfigure[Dataset sensitivity] {\includegraphics[width=1.35 in]{figures/dataset_obj}}  
% \subfigure[Varying $\tau$] {\includegraphics[width=1.35 in]{figures/threshold_obj}}
% \subfigure[Maximum LCC on \textit{Flickr}] {\includegraphics[width=1.35 in]{figures/fl_obj_function}}  
% \subfigure[Time on \textit{Flickr}] {\includegraphics[width=1.35 in]{figures/fl_time_function}}
% \caption{Scalability \& sensitivity tests for NILD-M}
% %\vspace{-8pt}
% \label{fig:scalability}
% \vspace{-10pt}
% \end{figure*}

\begin{figure*}[t]
\centering
\subfigure[Maximal LCC on \textit{CPEP}] {\includegraphics[width=1.35 in]{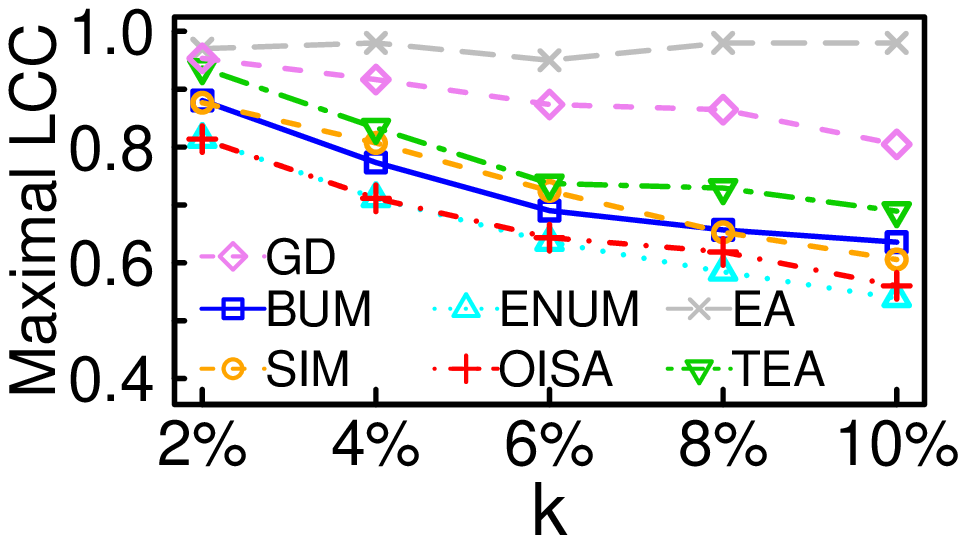}}  
\subfigure[Running time on \textit{CPEP} ]{\includegraphics[width=1.35 in]{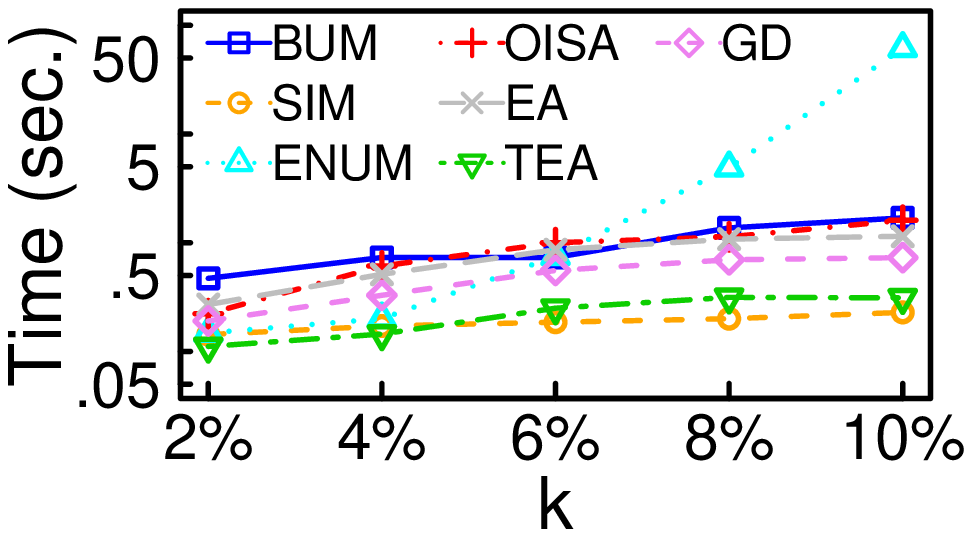}}  
\subfigure[LCC change on \textit{CPEP}]{\includegraphics[width=2.8 in]{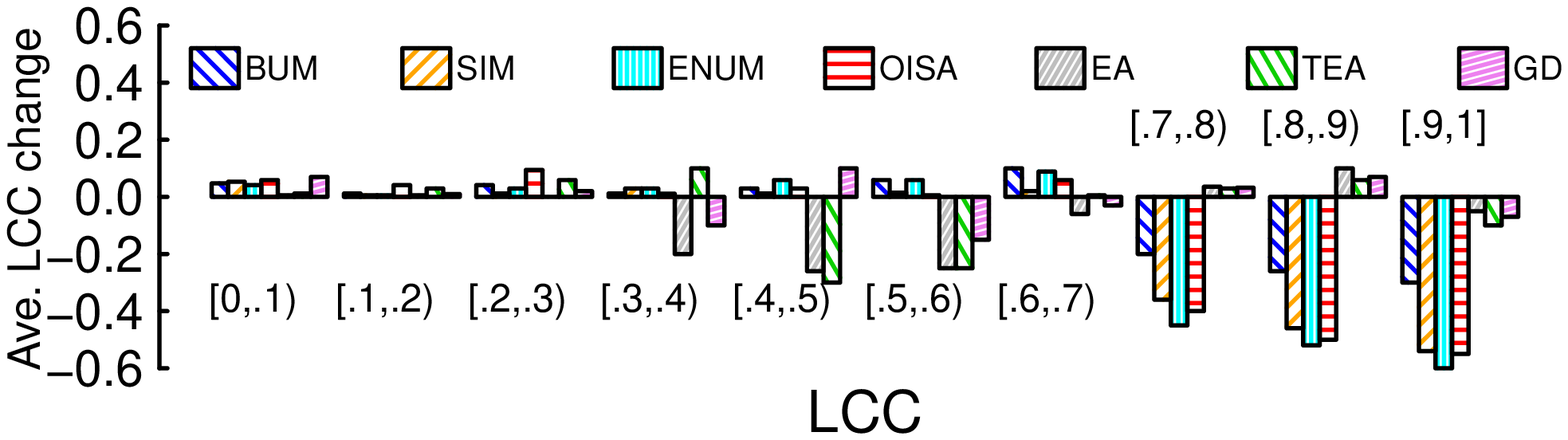}} 
\subfigure[Maximal LCC on \textit{Flickr}] {\includegraphics[width=1.35 in]{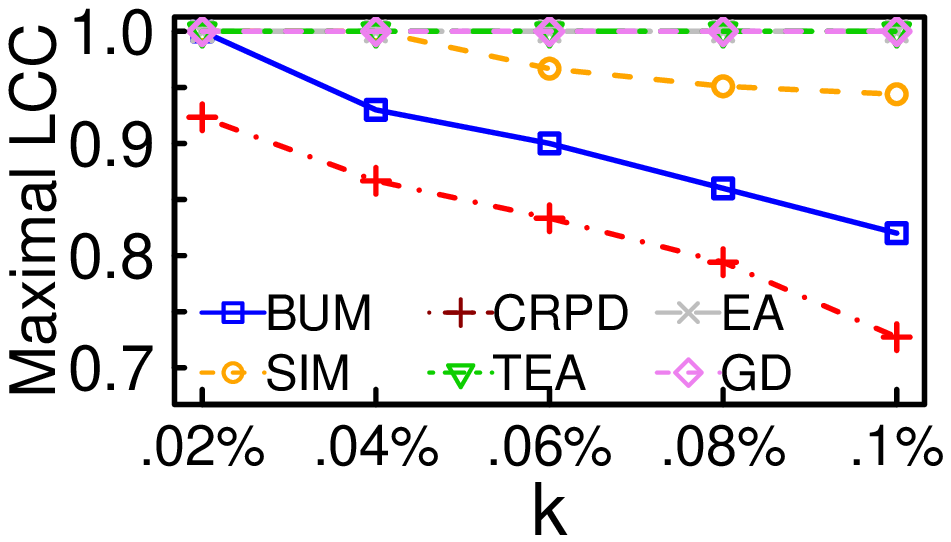}}
\subfigure[Time on \textit{Flickr}]{\includegraphics[width=1.35  in]{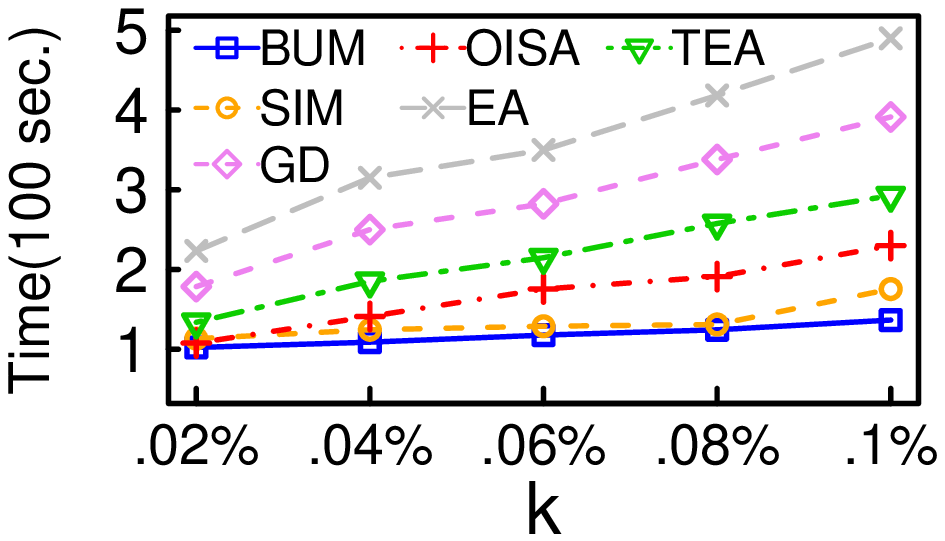}}
%\subfigure[LCC change on \textit{Flickr}]{\includegraphics[width=2.7 in]{figures/fl_lcc_distribution}} 
%\subfigure[Ave. betweenne of $T$ on \textit{Flickr}] {\includegraphics[width=1.35 in]{figures/nildm_fl_bet}}
\subfigure[Ave. closeness of $T$ on \textit{Flickr}]{\includegraphics[width=1.35 in]{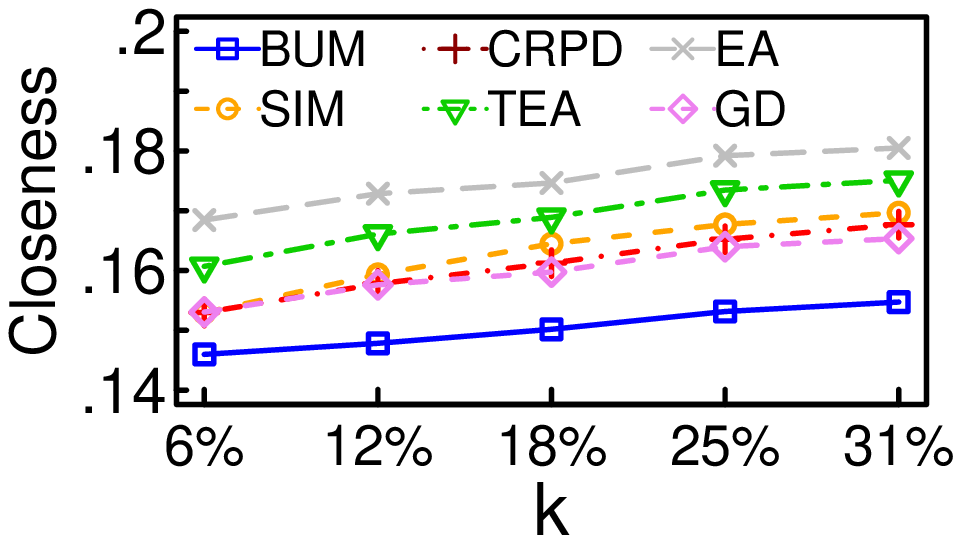}} 
\subfigure[Ave. LCC of datasets] {\includegraphics[width=1.35 in]{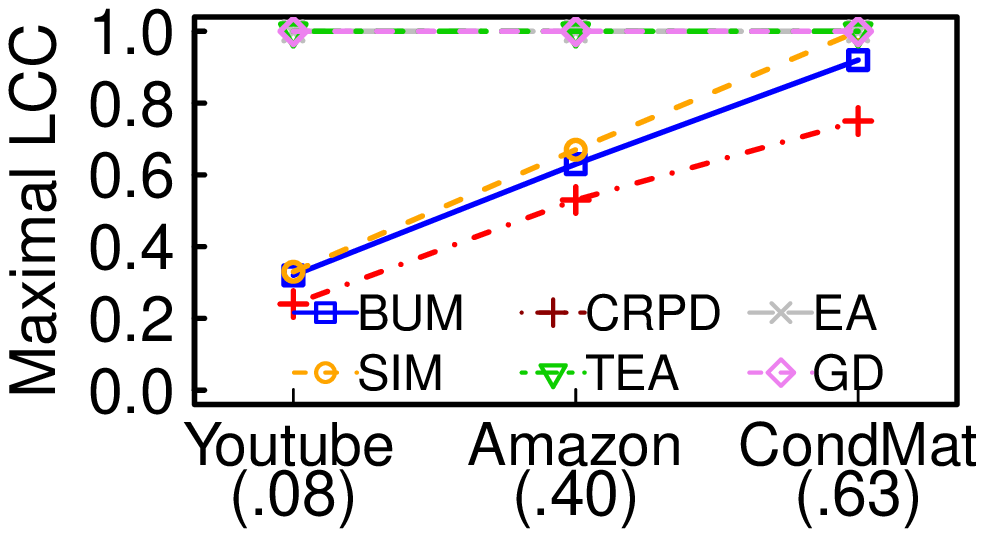}}  
%\subfigure[Varying $\tau$] {\includegraphics[width=1.35 in]{figures/threshold_obj}}
\subfigure[Maximum LCC on \textit{Flickr}] {\includegraphics[width=1.35 in]{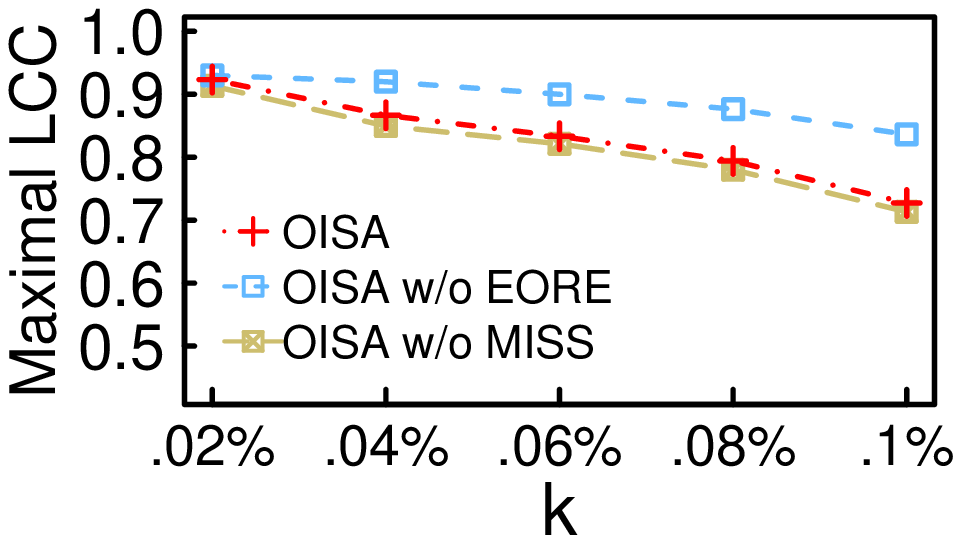}}  
\subfigure[Time on \textit{Flickr}] {\includegraphics[width=1.35 in]{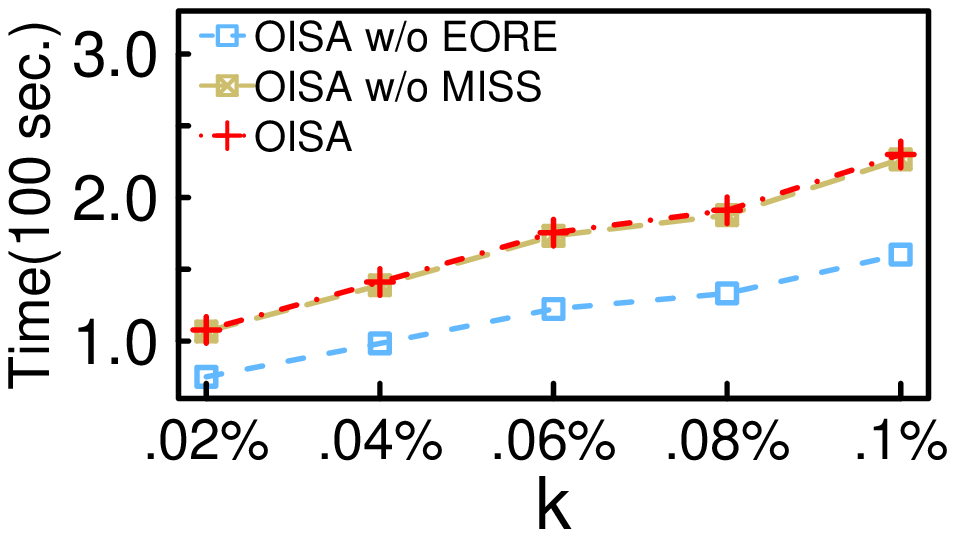}}
\caption{Scalability \& sensitivity tests for NILD-M}
%\vspace{-8pt}
\label{fig:scalability}
\vspace{-10pt}
\end{figure*}

%\vspace{-3pt}
%\subsubsection{Effectiveness of OISA in CPEP}
%\label{subsubsec:effectiveness}
Figures~\ref{fig:scalability}(a)-(c) first compare the effectiveness of all examined approaches on \emph{CPEP} by varying the number of intervention edges $k$ (relative to the edge number in $G$). %Although ENUM can obtain the optimal solution, it is not scalable since all combinations of $k$ intervention edges are required to be examined.
Figure~\ref{fig:scalability}(a) indicates that as $k$ grows, the maximum LCC generally decreases. 
%{\color{blue}
OISA and ENUM significantly outperform BUM, SIM, EA, TEA, and GD because EA, TEA and GD are designed to maximize the closeness centrality and influence score, instead of reducing the maximum LCC.%}
Also, BUM only examines the LCCs of the nodes, while SIM ignores LCCs and gives the priority to the intervention edges with the maximum numbers of hops. Figure~\ref{fig:scalability}(b) presents the running time of all approaches in the log-scale. OISA and most baselines select $F$ within 1 second. In contrast, the running time of ENUM grows exponentially.

To understand the changes of LCCs in different nodes, we take a closer look at the nodes whose LCC potentially changes, i.e., the terminal nodes of the selected edges $F$ and their neighbors.
Figure~\ref{fig:scalability}(c) presents the average LCC change of the nodes in each LCC range after intervention where $k=6\%$. %and $\tau=0.12$. 
As shown in Figure~\ref{fig:scalability}(c), even though BUM outperforms SIM, the average LCC reduction in the range [0.9, 1]
is smaller than SIM since BUM may connect nearby nodes and increase the LCCs of the common neighbors. It also explains why the maximum LCC achieved by BUM drops slower than SIM, ENUM, and OISA while $k$ increases, i.e., BUM creates lots of targeted nodes with LCCs around 0.6 and 0.7, and BUM needs to make intervention for all of them again to reduce the maximum LCC to become lower than 0.6. In contrast, the behavior of OISA is similar to ENUM in most ranges, and it successfully achieves comparable performance.

%\subsubsection{Scalability Tests}
%\vspace{-5pt}
%\vspace{-9pt}
%\subsubsection{Scalability of OISA}
%\label{subsubsec:scalability}

Next, Figures~\ref{fig:scalability}(d)-(i) compare all approaches except ENUM on \textit{Flickr}, since ENUM does not return any solution in two days even for $k = 10$. The result on \textit{Facebook} is similar and thereby is not shown here.
As there are nearly 9000 nodes with LCCs as 1 on \textit{Flickr}, the minimal $k$ is 4500 edges, because adding one edge can make intervention for at most two nodes with LCC as 1. 
Thus, $k$ is set to 0.02\%, 0.04\%, 0.06\%, 0.08\%, and 0.1\% of the number of edges (i.e., 4500, 9000, 13500, 18000, and 22500 edges). Figure~\ref{fig:scalability}(d) indicates that OISA significantly outperforms all other baselines under all settings of $k$. BUM is superior to SIM when $k$ is small because the farthest node $u$ selected by SIM usually results in a small LCC and fewer neighbors, i.e., reducing the LCC of $u$ does not help reduce the maximum LCC. Thus, when the number of intervention edges to be added is small, e.g., smaller than the number of targeted nodes in $T$, it is impossible for SIM to reduce the maximum LCC. In contrast, $u$ selected by BUM usually has a large LCC, and thus BUM is able to reduce the maximum LCC with the number of intervention edges around half of the minimal $k$.

\vspace{-1pt}
Figure~\ref{fig:scalability}(e) compares the running time. EA considers every possible edge and thus incurs the largest running time. 
BUM requires the least running time since it only retrieves the nodes with the largest LCC as the terminals of the intervention edges. OISA needs slightly more time but obtains much better solutions since it carefully examines multiple anticipated LCCs, i.e., $l_j$. Also, OISA takes a longer time on \textit{Flickr} since \textit{Flickr} has much more large-LCC nodes, and it is necessary for OISA to derive the optionality for all these nodes. The running time of EA and TEA grows significantly on \textit{Flickr} than on \textit{CPEP}, as EA, TEA and GD need to examine every candidate edge not in $E$, and \textit{CPEP} is denser than \textit{Flickr}. Figure~\ref{fig:scalability}(f) shows that the average closeness of targeted nodes in various $k$. Be noted that the average betweenness of targeted nodes has a similar trend with Figure~\ref{fig:scalability}(f) and thus eliminated here. The difference between OISA and baselines is similar to the case of NILD-S.

%\vspace{-7pt}
%\subsubsection{Sensitivity Tests} 
%\subsubsection{Parameter Sensitivity \& Component Effectiveness}
%\label{subsubsec:sensitivity}

Next, we conduct a sensitivity test on the average LCC of networks and show the component effectiveness. The results on \textit{CPEP} and \textit{Facebook} are similar to those of  \textit{Flickr} and thus not shown here.

%$\newline$
\noindent \textbf{Varying average LCC. }
We evaluate OISA on \textit{Youtube} (average LCC 0.08), \textit{Amazon} (average LCC 0.40), and \textit{Cond-Mat} (average LCC 0.63) with $k=0.3\%$ of the number of edges in each network, i.e., $k=8963$ for \textit{Youtube}, $k=2778$ for \textit{Amazon}, and $k=280$ for \textit{Cond-Mat}. Figure~\ref{fig:scalability}(g) shows that the maximum LCC is larger when the average LCC of the dataset grows,  %because a dataset with a larger average LCC (e.g., \textit{Cond-Mat}) includes more nodes with large LCCs, and it is more difficult to reduce the maximum LCC. 
and OISA outperforms other baselines. 
%For parameter $\tau$, Figure~\ref{fig:scalability}(g) indicates that with a larger $\tau$, more candidates can be considered for $m$, leading to better performance.

\noindent \textbf{Component effectiveness. }
Figures~\ref{fig:scalability}(h) and \ref{fig:scalability}(i) evaluate the different components in OISA.  
The result indicates that OISA achieves high maximal LCC with smaller computation time.

% \begin{table*}[t]
% \caption{Results comparing the over-time changes in negative emotions of the intervention and control groups}
% \footnotesize
% \begin{tabular}{|c|l|r|r|r|r|l|c|l|r|r|r|r|l|}
% \hline
% \scriptsize
%  &                   & Estimates & \begin{tabular}[c]{@{}r@{}}Std errors\end{tabular} & Deg. of freedom     & $t$-values &  $p$-values & &      & Estimates & \begin{tabular}[c]{@{}r@{}}Std errors\end{tabular} & Deg. of freedom      & $t$-values &  $p$-values \\ \hline
% \footnotesize
% \multirow{4}{*}{\textit{Depression}} & $\beta_0$ & 10.25 & 1.15 & 71.81  & 8.93  & <0.0001***  & \multirow{4}{*}{\textit{Anxiety}}& $\beta_0$       & 11.29    & 1.58 & 82.18  & 7.15    & <0.0001***    \\ \cline{2-7} \cline{9-14} 
% & $\beta_1$       & 0.29     & 1.62 & 71.81  & 0.18    & 0.8602 & &$\beta_1$       & -1.18    & 2.23 & 82.18  & -0.53   & 0.5984  \\ \cline{2-7} \cline{9-14}  
% & $\beta_2$       & 0.33     & 0.18 & 238.00 & 1.78    & 0.0769   & & $\beta_2$       & 0.61     & 0.32 & 238.00 & 1.90    & 0.0587       \\ \cline{2-7} \cline{9-14} 
% & $\beta_{3}$    & -0.74    & 0.21 & 238.00 & -3.59   & 0.0004*** & &$\beta_{3}$    & -0.83    & 0.36 & 238.00 & -2.30   & 0.0212*        \\ \hline
% \end{tabular}
% %\flushright p-value: <0.001 (***), < 0.01(**), < 0.05(*) 
% \label{table:user_study}
% \end{table*}

%\vspace{-5pt}
\subsection{Empirical Study}
\label{sec:human_study}
The empirical study aims at evaluating the utility and feasibility of the proposed network intervention algorithm in real-world settings. The study, spanned over two months, included 8 weekly measurements of psychological outcomes among 424 participants, aged between 18 and 25. 
The participants were university students and employees %at a university in Taiwan, 
with 638 pre-existing friendship links at the beginning of the study.
Four self-reported standard psychological questionnaires were adopted as indicators, % of health outcomes, 
including \textit{Beck Anxiety Inventory} (BAI)~\cite{beck1988inventory} for \textit{anxiety}, \textit{Perceived Stress Scale} (PSS)~\cite{cohen1994perceived} for \textit{stress}, \textit{Positive And Negative Affect Schedule} (PANAS)~\cite{watson1988development} for \textit{emotion}, and \textit{Psychological well-being Scale} (PWS)~\cite{sparq} for \textit{well-being}. Table~\ref{table:questionnaires} summarizes some evaluated items in the above questionnaires. 
In PANAS, there are 12 positive emotion terms and 14 negative emotion terms, and the overall score is the total score of 12 positive emotion terms minus the total score of negative emotion terms.  
For \textit{anxiety} and \textit{stress}, higher scores indicate higher levels of anxiety and perceived stress, respectively.  
For \textit{emotion} and \textit{well-being}, higher scores imply better emotion and psychological well-being.

%move to online version
\begin{table}[t]
\caption{Items of the adopted psychological questionnaires}
\label{table:questionnaires}
\footnotesize
\begin{tabular}{|p{0.25\columnwidth}|p{0.65\columnwidth}|}
\hline
Questionnaire & Sampled Items \\ \hline
%\begin{tabular}{l}
BAI (anxiety), 
33 items with 0--3 points
%(0: not at all, 1: mildly, 2: moderately, and 3: severely)
%\end{tabular} 
& 
\begin{itemize}[leftmargin=0.1in]
\vspace{-5pt}
    \item Numbness or tingling
    \item Unable to relax
    \item Fear of worst happening 
\vspace{-5pt}
\end{itemize}
\\ \hline
%\begin{tabular}{l}
PSS (Stress), 
5 sampled items with 0--4 points
%(0: never, 1: almost never, 2: sometimes, 3: fairly often, 4:very often) 
%\end{tabular}
& 
\begin{itemize}[leftmargin=0.1in]
    \vspace{-5pt}
    \item In the last month, how often have you been upset because of something that happened unexpectedly?
    \item In the last month, how often have you felt nervous and stressed?
    \item In the last month, how often have you felt that things were going your way?
    \vspace{-5pt}
\end{itemize}
\\ \hline
%\begin{tabular}{l}
PANAS (emotion),  
26 sampled items with 1--4 points 
%(1: never, 2: rarely, 3: often, 4: very often) 
%\end{tabular}
&
\begin{itemize}[leftmargin=0.1in]
    \vspace{-5pt}
    \item Positive emotions: excited, happy, satisfied, calm, relaxed, interested, proud, lively, love, determined, attentive, active
    \item Negative emotions: tired, depressed, sad, angry, anxious, distressed, upset, scared, irritable, ashamed, nervous, tense, hostile, afraid
    \vspace{-5pt}
\end{itemize}
\\ \hline
%\begin{tabular}{l}
PWS (well-being), 
18 items with 1--7 points 
%(1: strongly agree, 2:somewhat agree, 3: a little agree,  4: neither agree or disagree, 5: a little disagree, 6: somewhat disagree, 7:strongly disagree) 
%\end{tabular}
& 
\begin{itemize}[leftmargin=0.1in]
    \vspace{-5pt}
    \item I like most parts of my personality.
    \item When I look at the story of my life, I am pleased with how things have turned out so far.
    \item Some people wander aimlessly through life, but I am not one of them.
    \vspace{-5pt}
\end{itemize}
\\ \hline
\end{tabular}
\end{table}

%\todo[inline]{Include citations for the full scales, and give some examples of items in each scale.}

To evaluate the effects of adding friendship links based on different approaches, the participants were randomly assigned to one of the following four groups: three \textit{intervention groups} and one \textit{control group}, with 103 participants in every group.\footnote{Before the study started, participants in the four groups were requested to fill the questionnaires to ensure that there were no statistically significant differences for those questionnaires between every pair of groups.}
Participants in the intervention groups were provided with explicit friendship recommendations suggested by OISA and other two baselines, GD and BUM, respectively. Among the five baselines, GD was chosen because it can be applied to propagate health-related information to prevent obesity~\cite{wilder2018optimizing}.
%\textbf{[say the goal specifically and cite it]}.% health-related outcomes by adding new edges, even though its focus was preventative health but not psychological health. 
%wilder2018optimizing
BUM was chosen because it performs the best among all baselines in Section~\ref{subsec:oisa_evaluation}. In the study, $k$ was set as 28 for OISA, GD, and BUM, whereas $\tau = 0.12$, $\omega_{b} = 0.01$, $\omega_{c} = 0.1$, and $|T| =103$ for OISA. For OISA, GD, and BUM, the recommended edges were suggested by providing specific instructions for the participants to engage in online chatting. % with the recommended ``friends'' every week. 
In contrast, participants in the control group received no intervention and explicit instruction to interact with other participants. Each participant was required to provide responses to the four questionnaires every week for a total of eight times in this study.

An important question in the study is \textit{whether the participants accepted the friendship recommendations or not}. To answer this question, participants %who received friendship recommendations 
were also asked the following questions at the end of the study: \\
% \begin{enumerate}[label=\textit{Q{\arabic*}.}, leftmargin=1.5\parindent]
% \item \textit{Do you feel happy chatting with this recommended participant?}
% \item \textit{After the study ends, are you willing to chat with this participant?}
% \item \textit{If possible, are you willing to become a friend of this participant?}
% \end{enumerate}
\noindent \textit{Q1. Do you feel happy chatting with this recommended participant?}
\noindent \textit{Q2. After the study ends, are you willing to chat with this participant?}
\noindent \textit{Q3. If possible, are you willing to become a friend of this participant?}
For Q1, 80.2\% of the participants reported that they felt happy during the chat with the recommended participants. For Q2, 79.9\% of the participants replied that they would be willing to chat with the recommended participants even after the study ended. For Q3, 83.9\% of the participants reported that they would be willing to become a friend of the recommenced participant. Accordingly, participants in this study tended to accept the friend recommendations.

\begin{figure}[t]
\centering
\subfigure[Comparison with baselines] {\includegraphics[width=1.4 in]{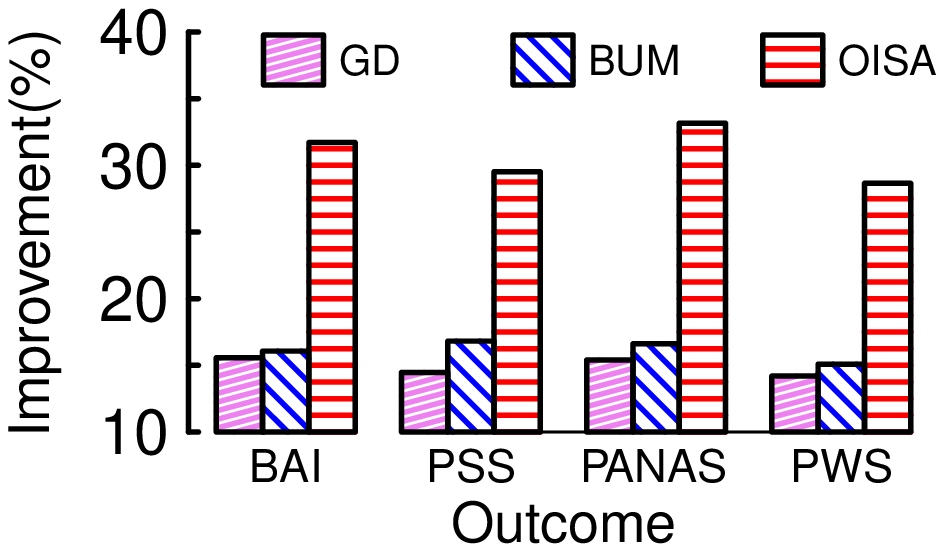}} 
\subfigure[Decrements in LCC (anxiety)] {\includegraphics[width=1.4 in]{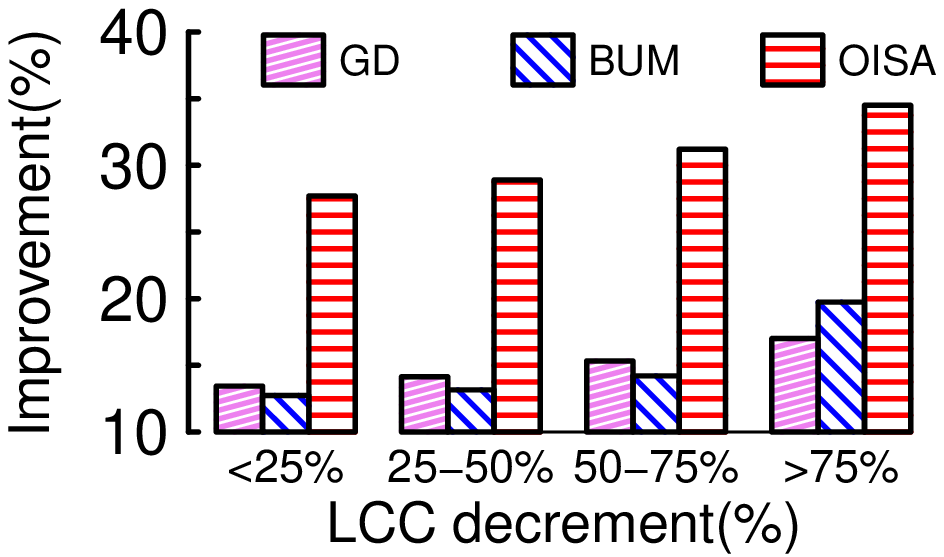}}  
%\subfigure[Decrements in LCC (PSS)]{\includegraphics[width=1.35 in]{figures/user_study_new_lcc_pss}}
\subfigure[Decrements in LCC (emotion)] {\includegraphics[width=1.4 in]{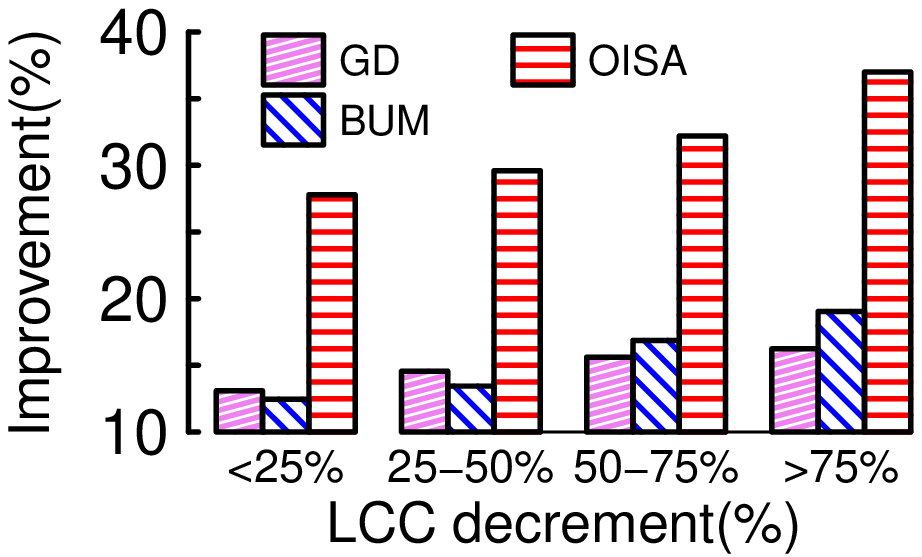}}  
\subfigure[Anxiety over time]{\includegraphics[width=1.4 in]{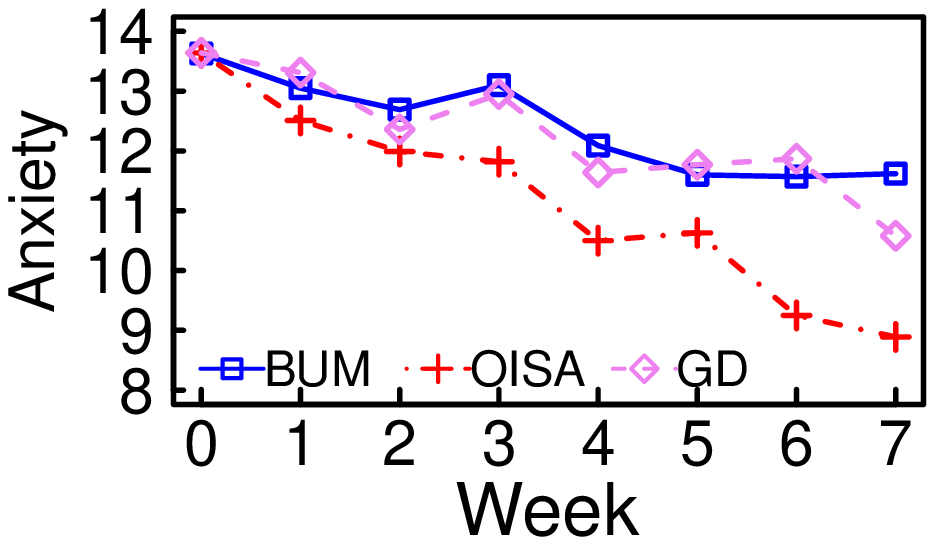}} 
\subfigure[Emotion over time]{\includegraphics[width=1.4 in]{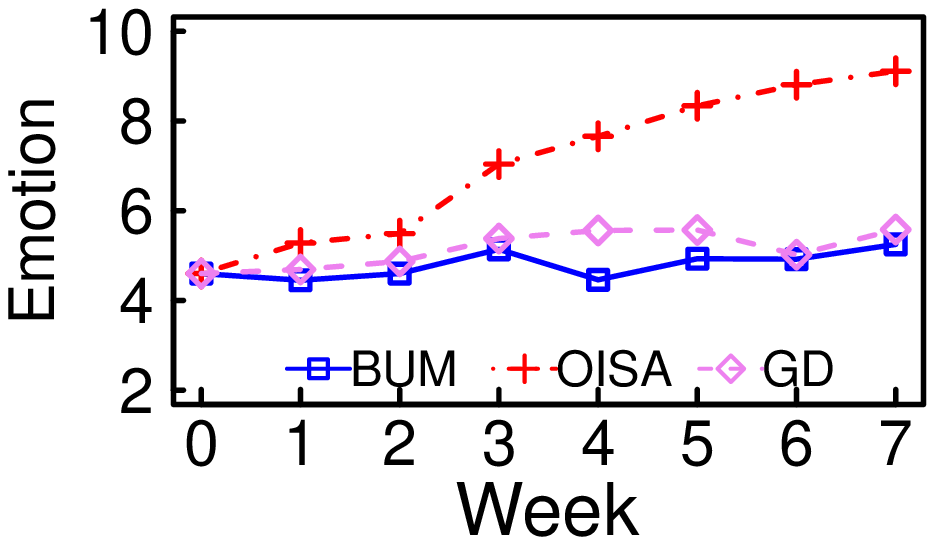}} 
\subfigure[Decrements in LCC]{\includegraphics[width=1.4 in]{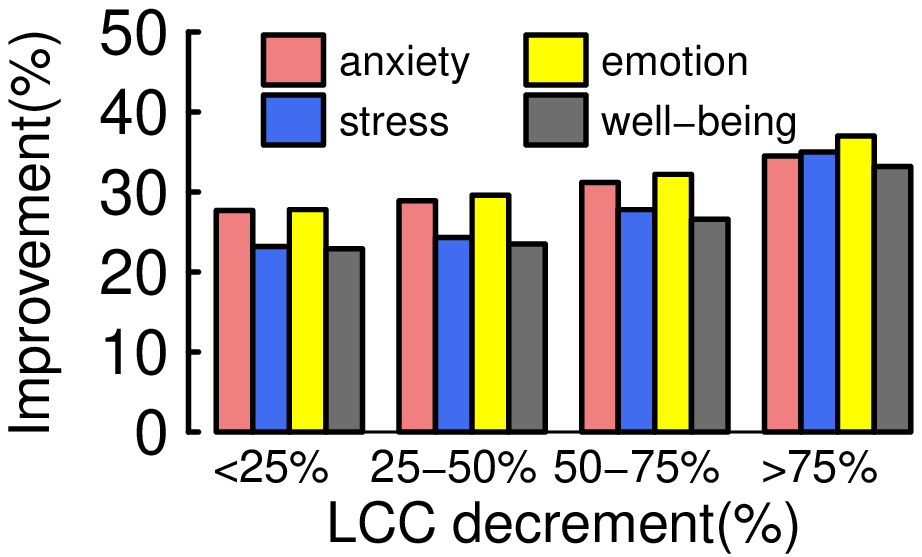}}
\vspace{-3pt}
\caption{Empirical study results}
%\vspace{-8pt}
\label{fig:user_study}
%\vspace{-10pt}
\end{figure}

To evaluate the effects of the intervention, Figure~\ref{fig:user_study} reports the average improvement on each psychological outcome for GD, BUM and OISA. %, from the start to the end of the study. 
As shown in Figure~\ref{fig:user_study}(a), among three intervention groups, OISA outperforms GD and BUM in all four measures, \textit{anxiety}, \textit{stress}, \textit{emotion}, and \textit{well-being}. Figures~\ref{fig:user_study}(b)-(c) further divide participants into four sub-groups according to the percentage of reduction in LCC.\footnote{Results of pressure, and well-being are similar and thus not shown here.} % and reports the average improvement on emotion for GD, BUM, and OISA. 
Figures~\ref{fig:user_study}(b)-(c) manifest that greater percentages of reduction in LCC are associated with more significant improvements in the psychological outcomes. It validates that new friendships via OISA is able to improve the psychological outcomes in this study. GD is inferior due to a different goal (i.e., maximizing the social influence score). %Even the objective of BUM is to optimize LCC, BUM recommend many friends of friends to a candidate, which is inappropriate.}  
Figures~\ref{fig:user_study}(d)-(e) plot the average scores of anxiety and emotion of participants in the intervention groups of GD, BUM, and OISA over time. As shown in Figures~\ref{fig:user_study}(d)-(e), the improvements of OISA is the most significant. %\textbf{because OISA introduces the suitable candidate new friends to participants whose ego networks are less diverse. 
In contrast, GD and BUM do not consider multiple network characteristics simultaneously. %and thus do not improve the outcomes significantly}.
%\textbf{because chatting with the friend candidate recommended by OISA shed a light to the participants by bringing interesting news, stories, etc.}.   
Figure~\ref{fig:user_study}(f), which plots the improvement of OISA for each outcome, also manifests that the improvements on anxiety, stress, and emotion with negative feelings are slightly better than that on well-being with only positive feelings, because brains tend to focus on potential threatening and negative emotions~\cite{kensinger2007negative}. In this case, the friend candidate recommended by OISA is able to provide social support to the participants.

We also evaluate OISA, GD and BUM separately with mixed effects modeling~\cite{jaeger2008categorical}, a statistical technique to examine if the intervention group and control group are statistically different.  Specifically, we fitted the model:
%\vspace{-3pt}
\small
\begin{align}
H_{it} =& (\beta_0 + \beta_{0i}) + \beta_1\cdot intervention_{i} + \beta_{2}\cdot time \\ \nonumber
&+ \beta_{3}\cdot intervention_{i}\cdot time  + \epsilon_{it},
\end{align}
%\vspace{-3pt}
\normalsize
where $H_{it}$ represents participant $i$'s emotion at time $t$; $intervention_i$ denotes whether participant $i$ is in the control group ($intervention_i=0$) or the intervention group ($intervention_i=1$); $time$ represents the number of weeks elapsed since the study started. The study starts at time 0. $\beta_0$ is a group intercept term representing the predicted outcome for the control group at time 0. $\beta_{0i}$ is the participant $i$'s deviation in intercept relative to $\beta_0$, which is a random effect in intercept at time 0.
%******
%\textbf{[add a footnote to explain what are the roles of the two terms, what to "intercept" means and add citations...]} 
$\beta_1$ --- $\beta_{3}$ are the regression weights associated with $intervention_i$, $time$ and their interaction, respectively. Specifically, $\beta_1$ indicates the difference in predicted psychological outcomes between the control group and intervention group at time 0; $\beta_2$ represents the estimated amount of change in emotion in the control group for each week of elapsed time (i.e., the ``placebo effect'', or the improvements in psychological outcomes shown by the control group). 
% ******
%\textbf{[why it is called placebo effect??? Does it indeed improve the outcome?]} % simply for being in the study). 
Finally, $\beta_3$ is the most important, i.e., the regression coefficient for the $intervention_i \cdot time$ interaction reveals the estimated difference in the amount of change in outcomes reported by the intervention group \textit{relative to} the control group for each week of elapsed time. %In other words, $\beta_3$ is the key term for testing our hypothesized effect of the OISA algorithm. 
If $\beta_3$ is statistically significantly different from 0, OISA (or other baselines) is able to change the participants' health outcomes with time substantially more than what is expected in the control group. %\footnote{Note that we only included person-specific deviation in the intercept term from $\beta_0$, but not other person-specific deviations in $\beta_1$ ---$\beta_3$ due to the relatively small sample size of the study, and the fact that the variances of these other random effects were minuscule and not substantially different from 0.} 
Finally, $\epsilon_{it}$ represents the residual error in the negative emotion that cannot be accounted for by other terms.

Table~\ref{table:user_study} presents the results of model fitting of OISA.
%***
%\textbf{and [where are others...]}. 
Estimates with $p$-values smaller than 0.05 are identified as statistically significantly different from 0. Thus, the two values of $\beta_0$, the predicted psychological outcomes of the control group at time 0, are estimated to be significantly different from 0. The term $\beta_1$ is not significantly different from 0 for all the four outcomes, which validates our random assignment procedure --- indicating that the intervention and control group do not differ substantially in their anxiety, stress, emotion and well-being scores at time 0. Also, the control group does not show substantial changes in all the four outcomes with time, as $\beta_2$ is not significant. Moreover, the values of $\beta_{3}$ are negative and significantly different from zero for anxiety and stress, indicating that the intervention group shows statistically greater decreases in anxiety and stress. Similarly, the values of $\beta_{3}$ are positive and significantly different from zero for emotion and well-being, indicating that the intervention group shows statistically greater increase in emotion and well-being. In other words, the intervention group, whose friendship recommendation is suggested by OISA, shows statistically greater improvement for all the four outcomes.
Also, the more extreme test statistic value ($t$-value) and its corresponding $p$-value, associated with the outcomes, suggest that the intervention have a more systematic effect (i.e., less uncertainty or smaller standard error)~\cite{casella2002statistical}. 
Also, Table~\ref{table:baseline_gd} and Table~\ref{table:baseline_bum} reveal that the model fitting results for BUM and GD are not statistically significant.%\footnote{The standard errors, degree of freedom, and $t$-values are shown in~\cite{online_version}.} 

%***
%\textbf{[You need to show at least the estimates of them! done)

\begin{table}[t]
\caption{Results comparing the over-time changes in outcomes of the intervention (OISA) and control groups}
\footnotesize
\begin{tabular}{|p{0.5in}|l|r|r|r|r|l|}
\hline
\scriptsize
 &                   & Estimates & \begin{tabular}[c]{@{}r@{}}Std errors\end{tabular} & df     & $t$-values &  $p$-values \\ \hline
\footnotesize
\multirow{4}{*}{\textit{Anxiety}}        
& $\beta_0$       & 13.64    & 1.17 & 320.64  & 11.70    & <0.0001***            \\ \cline{2-7} 
& $\beta_1$       & -0.69    & 1.65 & 323.10  & -0.42    & 0.6763                 \\ \cline{2-7} 
& $\beta_2$       & -0.09    & 0.15 & 1082.86 & -0.59    & 0.5532           \\ \cline{2-7} 
& $\beta_{3}$     & -0.61    & 0.21 & 1085.94 & -2.97    & 0.0030**         \\ \hline
\multirow{4}{*}{\textit{Stress}} 
& $\beta_0$       & 14.84    & 0.32 & 364.78  & 46.02    & <0.0001***              \\ \cline{2-7} 
& $\beta_1$       &  0.76    & 0.46 & 367.97  &  1.66    & 0.0982           \\ \cline{2-7} 
& $\beta_2$       &  0.07    & 0.04 & 1083.49 &  1.56    & 0.1187            \\ \cline{2-7} 
& $\beta_{3}$     & -0.19    & 0.06 & 1087.00 & -2.95    & 0.0033**            \\ \hline
\multirow{4}{*}{\textit{Emotion}}           
& $\beta_0$       &  4.60    & 1.14  & 298.29  &  4.05   & <0.0001***              \\ \cline{2-7} 
& $\beta_1$       & -0.08    & 1.61  & 300.36  & -0.05   & 0.9580           \\ \cline{2-7} 
& $\beta_2$       & -0.25    & 0.13  & 1082.17 & -1.85   & 0.0640            \\ \cline{2-7} 
& $\beta_{3}$     &  0.96    & 0.19  & 1084.98 &  5.09   & <0.0001***            \\ \hline
\multirow{4}{*}{\textit{Well-being}}           
& $\beta_0$       & 74.54    & 1.07 &  241.65  &  69.73   & <0.0001***              \\ \cline{2-7} 
& $\beta_1$       &  0.09    & 1.51 &  242.67  &   0.06   & 0.9545           \\ \cline{2-7} 
& $\beta_2$       & -0.18    & 0.09 & 1081.12  &  -1.95   & 0.0515            \\ \cline{2-7} 
& $\beta_{3}$     &  0.45    & 0.13 & 1082.87  &   3.41   & 0.0007***            \\ \hline
\end{tabular}
%\flushright p-value: <0.001 (***), < 0.01(**), < 0.05(*) 
\label{table:user_study}
\end{table}

\begin{table}[t]
\caption{Results comparing the over-time changes in outcomes of the intervention (GD) and control groups}
\footnotesize
\begin{tabular}{|p{0.5in}|l|r|r|r|r|l|}
\hline
\scriptsize
 &                   & Estimates & \begin{tabular}[c]{@{}r@{}}Std errors\end{tabular} & df     & $t$-values &  $p$-values \\ \hline
\footnotesize
\multirow{4}{*}{\textit{Anxiety}}        
& $\beta_0$       & 13.64    & 1.29 & 293.57  & 10.61    & <0.0001***            \\ \cline{2-7} 
& $\beta_1$       & -0.48    & 1.88 & 307.70  & -0.26    & 0.7990                 \\ \cline{2-7} 
& $\beta_2$       & -0.09    & 0.16 & 972.57  & -0.56    & 0.5780           \\ \cline{2-7} 
& $\beta_{3}$     & -0.29    & 0.24 & 981.08  & -1.23    & 0.2190           \\ \hline
\multirow{4}{*}{\textit{Stress}} 
& $\beta_0$       & 14.84    & 0.32 & 312.83  & 45.71    & <0.0001***              \\ \cline{2-7} 
& $\beta_1$       &  0.09    & 0.48 & 329.18  &  0.20    & 0.8440           \\ \cline{2-7} 
& $\beta_2$       &  0.07    & 0.04 & 973.09  &  1.66    & 0.0970           \\ \cline{2-7} 
& $\beta_{3}$     &  0.07    & 0.06 & 982.44  &  1.16    & 0.2450            \\ \hline
\multirow{4}{*}{\textit{Emotion}}           
& $\beta_0$       &  4.60    & 1.19  & 267.38  &  3.84   & 0.0002***              \\ \cline{2-7} 
& $\beta_1$       &  0.36    & 1.75  & 278.62  &  0.20   & 0.8394           \\ \cline{2-7} 
& $\beta_2$       & -0.25    & 0.13  & 969.85  & -1.85   & 0.0648            \\ \cline{2-7} 
& $\beta_{3}$     &  0.38    & 0.20  & 977.27  &  1.86   & 0.0638            \\ \hline
\multirow{4}{*}{\textit{Well-being}}           
& $\beta_0$       & 74.54    & 1.13 &  228.79  &  65.74   & <0.0001***              \\ \cline{2-7} 
& $\beta_1$       & -0.35    & 1.65 &  234.92  &  -0.21   & 0.8322           \\ \cline{2-7} 
& $\beta_2$       & -0.18    & 0.10 &  970.51  &  -1.86   & 0.0636            \\ \cline{2-7} 
& $\beta_{3}$     &  0.25    & 0.15 &  975.22  &   1.67   & 0.0946            \\ \hline
\end{tabular}
%\flushright p-value: <0.001 (***), < 0.01(**), < 0.05(*) 
\label{table:baseline_gd}
\end{table}

 \begin{table}[t]
\caption{Results comparing the over-time changes in outcomes of the intervention (BUM) and control groups}
\footnotesize
\begin{tabular}{|p{0.5in}|l|r|r|r|r|l|}
\hline
\scriptsize
 &                   & Estimates & \begin{tabular}[c]{@{}r@{}}Std errors\end{tabular} & df     & $t$-values &  $p$-values \\ \hline
\footnotesize
\multirow{4}{*}{\textit{Anxiety}}        
& $\beta_0$       & 13.63    & 1.21 &  327.22  & 11.24    & <0.0001***            \\ \cline{2-7} 
& $\beta_1$       & -3.08    & 1.73 &  333.55  & -1.79    & 0.0752                 \\ \cline{2-7} 
& $\beta_2$       & -0.09    & 0.16 & 1039.23  & -0.55    & 0.5827           \\ \cline{2-7} 
& $\beta_{3}$     & -0.26    & 0.23 & 1049.23  & -1.12    & 0.2615           \\ \hline
\multirow{4}{*}{\textit{Stress}} 
& $\beta_0$       & 14.84    & 0.33 &  342.38  & 45.46    & <0.0001***              \\ \cline{2-7} 
& $\beta_1$       &  0.86    & 0.46 &  349.29  &  1.85    & 0.0654           \\ \cline{2-7} 
& $\beta_2$       &  0.07    & 0.04 & 1040.66  &  1.60    & 0.1105           \\ \cline{2-7} 
& $\beta_{3}$     & -0.05    & 0.06 & 1051.23  & -0.76    & 0.4454            \\ \hline
\multirow{4}{*}{\textit{Emotion}}           
& $\beta_0$       &  4.60    & 1.11  &  299.82  &  4.15   & <0.0001***              \\ \cline{2-7} 
& $\beta_1$       &  2.42    & 1.58  &  304.93  &  1.53   & 0.1262           \\ \cline{2-7} 
& $\beta_2$       & -0.25    & 0.13  & 1039.63  & -1.87   & 0.0613            \\ \cline{2-7} 
& $\beta_{3}$     &  0.29    & 0.19  & 1048.12  &  1.52   & 0.1291            \\ \hline
\multirow{4}{*}{\textit{Well-being}}           
& $\beta_0$       & 74.54    & 1.08 &  240.07  &  69.00   & <0.0001***              \\ \cline{2-7} 
& $\beta_1$       &  1.00    & 1.53 &  242.46  &   0.65   & 0.5161           \\ \cline{2-7} 
& $\beta_2$       & -0.18    & 0.09 & 1038.35  &  -1.94   & 0.0529            \\ \cline{2-7} 
& $\beta_{3}$     &  0.26    & 0.14 & 1043.01  &   1.93   & 0.0537            \\ \hline
\end{tabular}
%\flushright p-value: <0.001 (***), < 0.01(**), < 0.05(*) 
\label{table:baseline_bum}
\end{table}

Lastly, inspection of these results by 11 clinical psychologists and professors\footref{fnlabel}, is carried out to observe the behavioral implications behind the scores. The psychologists and professors are asked to the following question in Likert scale: Is the network intervention help improve the participants' outcomes? Comparing their evaluation among intervention and control groups, the results indicate that 82\% of psychologists and professors agree that the recommendation of OISA is the most effective; while only 9\% of them agree for GD and BUM. 
%\textbf{Specifically, the psychologists and professors are asked to evaluate the level of all participants' improvements in Likert scale. The results indicate that the intervention suggested by OISA has 51\% improvement from the control group, outperforming the intervention suggested by GD and BUM, which have 26\% and 25\% improvements from the control group.}
The above results lead to consistent conclusion with the study --- the intervention is therapeutic and positive enough to be brought into clinical consideration. 
%***
%\textbf{[Plot some figures....Otherwise, this part will be useless!!!]}

%\todo[inline]{Hui-Ju: Please don't compare $p$-values or magnitudes of the raw regression coefficients between depression and anxiety directly. One could change more than the other simply due to scaling differences -- one point of decrease may not mean the same thing for depression than anxiety.  Also, the coefficient for depression is not necessarily higher in magnitude. You can, however, explain the way that I have above.}

%\vspace{-7pt}
\section{Conclusion}
\label{sec:conclusion}

Even though research has suggested the use of network intervention for improving psychological outcomes, there is no effective planning tool for practitioners to select suitable intervention edges from a large number of candidate network characteristics. In this paper, we formulate NILD-S and NILD-M to address this practical need. We prove the NP-hardness and inapproximability, and propose effective algorithms for them. %CRPD and OISA to solve NILD-S and NILD-M, respectively. 
Experiments based on real datasets show that our algorithms outperform other baselines in terms of both efficiency and effectiveness; empirical results further attest to the practical feasibility and utility of using the OISA algorithm for real-world intervention purposes.
%\vspace{-10pt}

\begin{acks}
This work was supported in part by the National Science Foundation (IIS-1717084, SMA-1360205, IGE-180687, SES-1357666), the Ministry of Science and Technology (107-2221-E-001-011-MY3, 108-2221-E-001-002, 108-2636-E-007-009-), the National Institutes of \\ Health (U24AA027684), the Penn State Quantitative Social Sciences Initiative, and the the National Center for Advancing Translational Sciences (UL TR000127).
\end{acks}

% \begin{acks}
% This work was supported in part under grants MOST 107-2221-E-001-011-MY3, MOST 108-2221-E-001-002, and MOST 108-2636-E-007-009.
% \end{acks}

%\vspace{-3pt}
%\newpage

%\newpage
%\textbf{TO BE MOVED TO APPENDIX}
\appendix

% \linespread{0.98}
\section{Appendices}

\subsection{Proof of Theorem \ref{thm:crpd_optimal}}
\label{proof:crpd_optimal}

Without loss of generality, we assume that nodes are numbered in the ascending order of their weights, i.e., $f(v_i) \geq f(v_j)$ if $i \geq j$. In the following, we first introduce the properties of threshold graphs regarding LCC in Corollary~\ref{lemma:threshold_graph_basic} and Lemma~\ref{lemma:lcc_update}, and then discusses the feasibility of the solution returned by CRPD in Lemma~\ref{lemma:crpd_feasible}.

%\begin{lemma}
\begin{corollary}
Given a threshold graph $G=(V,E)$ with the node weights $f(\cdot)$ and threshold $t_G$, the nodes (with ID 1, 2,...., $n$)are separated into three sets, $V_Z=\{i| 1 \leq i \leq z\}$, $V_D=\{i| z+1 \leq i \leq c\}$, and  $V_C = \{i| c+1 \leq i \leq n\}$, where $z$ is the largest node satisfying $f(z) + f(n) \leq t_G$, and $c$ is the node satisfying $f(c-1) + f(c) \leq t_G$ and $f(c) + f(c+1) > t_G$. Then, 1) any node $v \in V_Z$ has no incident edge. 2) For any two nodes $u$ and $v$ in $V_D$, $(u, v) \notin E$. 3) For any two nodes $u$ and $v$ in $V_C$, $(u, v) \in E$.
\label{lemma:threshold_graph_basic}
\end{corollary}
%\end{lemma}
% \begin{proof}
% We first prove (1) by contradiction. Assume that there is at least one edge $(v, \overline{v}) \in E$ with $v \in V_Z$ and $\overline{v} \in V$. Since $v \in V_Z$ and $\overline{v} \in V$, $f(v) \leq f(v_z)$ and $f(\overline{v}) \leq f(v_n)$ hold. Then, $f(v) + f(\overline{v}) \leq f(v_z) + f(v_n) \leq t_G$, leading to a contradiction to Definition~\ref{def:threshold_graph}. Next, we prove (2) by contradiction. Assume that there exists at least a pair of $\overline{v}_1$ and $\overline{v}_2$ in $V_D$ with $(\overline{v}_1, \overline{v}_2) \in E$. 
% According to Definition~\ref{def:threshold_graph}, $f(\overline{v}_1)+ f(\overline{v}_2) > t_G$ holds. However, since $v_c$ has the largest weight among all nodes in $V_D$, $f(\overline{v}_1)+ f(\overline{v}_2) \leq 2f(v_c) \leq f(v_c) + f(v_{c+1}) \leq t_G$, leading to a contradiction. Finally, we prove (3), also by contradiction. Assume that there exists $\overline{v}_1$ and $\overline{v}_2$ in $V_C$ such that $(\overline{v}_1, \overline{v}_2) \notin E$. It implies $f(\overline{v}_1)+ f(\overline{v}_2) \leq t_G$ according to Definition~\ref{def:threshold_graph}. However, since $\overline{v}_1 \in V_C$ and $\overline{v}_2 \in V_C$, $f(\overline{v}_1) +f(\overline{v}_2) \geq$ $2f(v_{c+1}) \geq f(v_{c+1})+ f(v_{c}) > t_G$, leading to a contradiction. 
% \end{proof}

\begin{lemma}
After $(i,j)$ is added to $G=(V,E)$, the LCC of every $v \in (V \backslash \{i, j\}) \backslash (N_G(i) \cap N_G(j))$ remains the same.
\label{lemma:lcc_update}
\end{lemma}
% \begin{proof}
% We prove this lemma in Appendix~\ref{proof:lcc_update}
% \end{proof}
\vspace{-2pt}
\begin{proof}
We also prove the lemma by contradiction. Assume the LCC of a $\overline{v} \in (V \backslash \{i, j\}) \backslash (N_G(i) \cap N_G(j))$ becomes different. Two possible cases exist. (1) $\overline{v}$ is not a neighbor of $i$ and $j$, i.e., $(\overline{v} , i) \notin E$ and $(\overline{v}, j) \notin E$. (2) $\overline{v}$ is the neighbor of either $i$ or $j$. Without loss of generality, we assume $(\overline{v} , i) \in E$ but $(\overline{v} ,j) \notin E$. For the first case, $i$ and $j$ are not in the subgraph induced by $\overline{v} $ and the neighbors of $\overline{v}$. Thus, adding $(i,j)$ does not change $\overline{v}$'s LCC. For the second case, $i$ is in the subgraph induced by $\overline{v}$ and its neighbors, but $j$ is not. Therefore, $(i,j)$ is not included in the induced subgraph. 
\end{proof}

%The following lemma discusses the feasibility of the solution returned by CRPD. 

\begin{lemma}
For $G=(V,E)$, if there exists at least one feasible solution following the constraints of NILD-S on $\tau$, $\omega_{b}$, $\omega_{c}$, and $\omega_{d}$, the solution obtained by CRPD is always feasible. 
\label{lemma:crpd_feasible}
\end{lemma}
% \begin{proof}
% We prove this lemma in Appendix~\ref{proof:crpd_feasible}.
% \end{proof}
%\vspace{-2pt}
\begin{proof}
Assume that $G=(V,E)$ has at least one feasible solution $V_{FS}$, but the solution $V_F$ obtained by CRPD is different and infeasible. If there are multiple feasible solutions, let $V_{FS}$ be the one with the most common nodes with $V_F$. 
Let $u_{fs}$ and $u_f$ denote the first different node in $V_{FS}$ and $V_F$ when all nodes are sorted according to their IDs in the threshold graph, respectively. Note that the ID of $u_f$ is smaller than $u_{fs}$; otherwise CRPD would choose $u_{fs}$ in $V_F$. In the following, we prove that connecting $t$ to the nodes in $V_{FF} = V_{FS}\backslash\{u_{fs}\}\cup\{u_f\}$ leads to another feasible solution following the constraints on $\omega_{b}$, $\omega_{c}$, and $\omega_{d}$. 
1) Both $V_{FS}$ and $V_{FF}$ add $k$ new neighbors to $t$, and the degree of $t$ for $V_{FF}$ always exceeds $\omega_{d}$. 2) The betweenness of $t$ is the proportion of shortest paths among all node pairs passing through $t$. If both $u_{fs}$ and $u_f$ are in $V_D \cup V_C$, the betweenness of $t$ is the same in $V_{FS}$ and $V_{FF}$, because the length of shortest paths among all nodes in $V_D \cup V_C$ is at most 2, and the betweenness will not be improved by edges in $\{(t, v)| v \in V_{FS}\}$ or $\{(t, v)| v \in V_{FF}\}$. 
If both $u_{fs}$ and $u_f$ are in $V_Z$, the betweenness of $t$ is the same, because the number of node pairs $\langle u_{fs}, v \rangle$ with $v \in V_D \cup V_C \backslash \{t\}$ and the corresponding shortest paths not passing through $t$ after removing $(t, u_{fs})$ is identical to the number of node pairs $\langle u_f, v \rangle$ with $v \in V_D \cup V_C \backslash \{t\}$ and the corresponding shortest paths passing through $t$ after adding $(t, u_f)$.
 If $u_{fs} \in V_D \cup V_C$ and $u_f \in V_Z$, the betweenness of $t$ grows and becomes larger than $\omega_{b}$, because the shortest paths of node pairs $\langle u_f, v \rangle$ with $v \in V_D \cup V_C \backslash \{t\}$ pass through $t$. 
Therefore, the betweenness of $t$ in solution $V_{FF}$ is larger than $\omega_{b}$. 3) The closeness of $t$ is the sum of reciprocal of distances to other nodes. If both $u_{fs}$ and $u_f$ are in $V_D \cup V_C$ (or both $u_{fs}$ and $u_f$ are in $V_Z$), the closeness of $t$ is the same. If $u_{fs} \in V_D \cup V_C$ and $u_f \in V_Z$, the closeness of $t$ increases because the distance between $t$ and $u_f$ changes from 1 to 3, but the distance between $t$ and $u_f$ changes from $\infty$ to 1.
\end{proof}

%Theorem~\ref{thm:crpd_optimal} proves that the feasible solution returned by CRPD is optimal.

Theorem~\ref{thm:crpd_optimal} proves that the above feasible solution is optimal.

%Finally, We Theorem~\ref{thm:crpd_optimal} proves that the feasible solution returned by CRPD is optimal. The proof is detailed in the following.

\setcounter{theorem}{1}
\begin{theorem}
The CRPD algorithm can find the optimal solution of NILD-S when $G=(V,E)$ is a threshold graph, and the running time of CRPD is $O(|E|\hat{d} + k\hat{d}|V|)$.
%\label{thm:crpd_optimal}
\end{theorem}
\setcounter{theorem}{4}
\vspace{-5pt}
\begin{proof}
%***check the notation later***
We prove the optimality by exploring the following two cases: 1) $|V_Z \cup V_D \backslash \{t\}| \geq k$ and 2) $|V_Z \cup V_D \backslash \{t\}| < k$. 
%Assume that there is a optimal solution $S*$ better than the solution obtained by [algorithm-name]. 
For the first case, according to Definition~\ref{def:threshold_graph}, a node with a larger weight is more inclined to connect to others. According to Corollary~\ref{lemma:threshold_graph_basic}, $V_Z \cup V_D$ forms an independent set. Also, the nodes are ordered in the ascending order of the degrees. Thus, CRPD examines nodes in the ascending order of their weights and chooses all nodes in $V_Z \cup V_D$. Adding edges in $\{(t, v)| v \in V_Z \cup V_D\}$ will optimize the LCC of $t$ because of the following reasons. (a) Adding edges in $\{(t, v)| v \in V_Z \cup V_D\}$  only changes the LCCs of the nodes in $\{t\} \cup V_Z \cup V_D$ according to Lemma~\ref{lemma:lcc_update}. (b) For the nodes in $V_Z$, its degree is increased by 1 (i.e, connecting to $t$), but the LCC remains as 0 because $V_Z \cup V_D$ is an independent set. (c) The LCCs of nodes in $V_D$ is 1, and the LCCs of nodes in $V_D$ are impossible to be increased. (d) The LCC of $t$ will be optimized because adding edges in $\{(t, v)| v \in V_Z \cup V_D\}$ only increases the denominator in Equation~\ref{eq:lcc}, but it does not increase the numerator in Equation~\ref{eq:lcc}. Thus, if $|V_Z \cup V_D \backslash \{t\}| \geq k$, CRPD finds the optimal solution. 

Next, we prove the second case by contradiction. Assume that there is an optimal solution $F^*=\{(t, v)| v \in V^*_F\}$ better than the solution $F=\{(t, v)| v \in V_F\}$ obtained by CRPD. When there are multiple optimal solutions, let $F^*=\{(t, v)| v \in V^*_F\}$ be the one such that $V^*_F$ has the most common nodes with $V_F$. Let $V_F = \{u_1, u_2, ..., u_{i-1}, u_{i}...u_k\}$ and $V^*_F=\{u_1, ...,u_{i-1}, u^*_{i}, ..., u^*_k\}$, such that $u_{i}$ and $u^*_{i}$ are the first different node in $V_F$ and $V^*_F$ when all nodes are sorted according to their IDs. $u_{i}$ is smaller than $u^*_i$; otherwise CRPD would select $u^*_i$ in $V_F$. Moreover, $N_{G} (u_i)\subseteq N_{G} (u^*_i)$ according to Definition~\ref{def:threshold_graph}. We consider another solution $\{(t, v)| v \in V^{**}_F\}$ with $V^{**}_F = V^*_F\backslash\{u^*_{i}\}\cup\{u_{i}\}$. First, the solution $\{(t, v)| v \in V^{**}_F\}$ is a feasible solution due to the following reason. If $u_{i} \in V_Z \cup V_D \backslash \{t\}$, the LCC of $u_{i}$ after adding $\{(t, v)| v \in V^{**}_F\}$ satisfies the LCC degradation constraint $\tau$ since the LCC of $u_{i} \in V_D$ is 1 before adding edges (i.e., never increase), and the LCC of $u_{i} \in V_Z$ remains as 0 before and after adding edges. On the other hand, if $u_{i} \in V_C$, the LCC of $u_{i}$ after adding $\{(t, v)| v \in V^{**}_F\}$ is identical to the one after adding $\{(t, v)| v \in V^{*}_S\}$, because all nodes $u_{i+1}...u_k$ and $u^*_{i+1}, ..., u^*_k$ are fully connected, and the number of edges between $u_{i}$'s neighbors is thereby the same. Second, the LCC of $t$ after adding $\{(t, v)| v \in V^{**}_F\}$ is not larger than the LCC of $t$ in solution $\{(t, v)| v \in V^{*}_F\}$, as $N_{G} (u_i)\subseteq N_{G} (u^*_i)$. Thus, it contradicts that $\{(t, v)| v \in V^*_F\}$ is the optimal solution with the most common nodes with $\{(t, v)| v \in V_F\}$. 

In LCC Calculation Step, to find the number of edges between node neighbors, 
CRPD stores a count for each node. For each edge $(i,j) \in E$, it examines the neighbors of the terminal nodes $N_G(i)$ and $N_G(j)$ in $O(\hat{d})$ time, where $\hat{d}$ is the maximum degree in $G=(V,E)$. For every common neighbor of $i$ and $j$, CRPD increases its count by 1. After examining all $|E|$ edges in $O(|E|\hat{d})$ time, CRPD extracts $n_v$ for each node $v$. (2) With $n_v$, CRPD finds $LCC_G(v)$ of every node in $G=(V,E)$ in $O(|V|)$ time, and the total time complexity in LCC Calculation Step is $O(|E|\hat{d} + |V|)$. In Edge Selection Process Step, CRPD first examines each of $t$'s two-hop neighbors $v$ and excludes $v$ if adding $(t, v)$ increases the LCC of any $t$'s one-hop neighbor to more than $\tau$ in $O(\hat{d}|V|)$ time. Then, CRPD acquires an initial solution in $k$ iterations. In each iteration, CRPD first finds $u$ from the remaining candidates in $O(|V|)$ time and excludes any node $v$ if adding $(t, v)$ will increase another node's LCC to more than $\tau$ in $O(|V|\hat{d})$ time. The total time to find the initial solution is $O(k\hat{d}|V|)$.  After that, CRPD explores another solution starting from $r$ in $O(k\hat{d}|V|)$ time. In summary, CRPD requires $O(|E|\hat{d} + |V|)$ in LCC Calculation Step and $O(k\hat{d}|V|)$ in Edge Selection Processing Step, and the total running time is $O(|E|\hat{d} + k\hat{d}|V|)$.\footnote{We follow existing algorithms to obtain the betweenness and closeness in an incremental manner, and the  worst-case time complexity is $O(|V||E|)$ to update the betweenness and closeness of all nodes after adding an edge~\cite{kas2013incremental,kas2013incremental2}. \label{complexity}}
\end{proof}

\subsection{Proof of Lemma~\ref{lemma:edge_lowerbound}}
\label{proof:edge_lowerbound}
%\begin{proof}
We prove the lemma by contradiction. Assume there exists a $k'_{t} < k_{t}$ such that the LCC of the node $t$ can be reduced to $l$ with only $k'_{t}$ edges. First, the LCC of $t$ becomes $\frac{n_{t} + n'_{t} }{C(d_{G}(t) + k'_{t},2)}$ after $k'_{t}$ edges are added and connected to $t$, where $n_{t}$ is the number of edges between $t$'s original neighbors before adding the new edges (i.e., $n_{t} = LCC_G(t)\times C(d_{G}(t), 2)$), and $n'_v$ is the number of edges between the original neighbors and the new neighbors or between any two new neighbors. Thus, it leads to a contradiction because
\small
\begin{equation}
\begin{aligned}
\frac{n_{t} + n'_{t} }{C(d_{G}(t) + k'_{t},2)} & \geq \frac{n_{t} }{C(d_{G}(t) + k'_{t},2)} \\
&= \frac{LCC_G(t)\times  d_{G}(t)(d_{G}(t)-1)}{(d_{G}(t) + k'_{t})(d_{G}(t)+k'_{t} -1)} 
> l,
\end{aligned}
\end{equation}
\normalsize
when $k'_{t} < k_{t}$. Note that $k_{t}=0$ if the LCC of $t$ is already smaller than $l$. Therefore, $0.5\times \sum_{t \in T} \max\{k_{t}, \omega_{d} - d_G(t)\}$ is a lower bound on $k_G$ for $l$ while satisfying the degree constraint $\omega_{d}$, since an edge connects two nodes.
%\end{proof}

% % \subsection{Proof of Theorem~\ref{lemma:upper_bound_relation}}
% % \label{proof:upper_bound_relation}
% % %\begin{proof}
% % %We prove the theorem by contradiction. Assume that $LCC_{G(V, E \cup F)}(v) > \overline{LCC}(v,k)$. 
% % %\textbf{why the follwoing is $k'_2$, instead of $k_2$? It may lead to ambiguity.}
% % Let $k_1$ and $k_2$ denote the numbers of intervention edges connecting to $v$ and any two neighbors of $v$, respectively. After intervention, $k_1 + k_2 \leq k$, and $LCC_{\overbar{G}}(v) = \frac{n_v + k_2 + y}{C(d_G(v) + k_1, 2)}$, where $y$ is the number of edges between the new neighbors via the $k_1$ new edges and the original neighbors of $v$ in $G$, and $y \leq d_G(v) \times k_1 + C(k_1, 2)$. Thus, $LCC_{\overbar{G}}(v) \leq \frac{n_v + k_2 + d_G(v) \times k_1 + C(k_1, 2)}{C(d_G(v) + k_1, 2)}$. Let $k_3 = k - k_1 - k_2$.
% % Then, we obtain $LCC_{\overbar{G}}(v)\leq$ $\frac{n_v + (k_2 +k_3)+ k_1 \times d_G(v) + C(k_1, 2) }{C(d_G(v) + k_1, 2)}$ \\
% % $\leq  \max_{k_1+k_2= k} \{ \frac{n_v + k_2 +k_1 \times d_G(v) + C(k_1, 2)}{C(d_G(v) + k_1, 2)}\}$  $=\overline{LCC}(v,k)$. 
% % %\end{proof}

\subsection{Proof of Theorem~\ref{thm:oisa_complexity}}
\label{proof:oisa_complexity}
%\begin{proof}
Before OISA starts, for every $l_j$, $k_G$ can be derived by solving a quadratic equation in $O(|T|)$ time. OISA skips an $l_j$ if $k_G > k$. Otherwise, PONF iteratively 1) retrieves the largest LCC among all nodes in $T$ in $O(|T|)$ time, 2) calculates the optionality of the nodes with the largest LCC to find $m$ in $O(n_m(|V|+|E|))$ time (where $n_m$ is the number of nodes with the largest LCC), and 3) chooses the node $u$ with the largest LCC among all option nodes to extract $u$ in $O(|T|)$ time. 
After $(m, u)$ is selected, ALC updates the LCCs in $O(\hat{d})$ time since ALC examines the LCCs of $m$, $u$, and their common neighbors, and updates their LCCs if necessary. The updates for $m$ and $u$ require $O(\hat{d})$ time, and the updates for the common neighbors take $O(1)$ time. 
Thus, the total running time is $O(n_l\times(|V| + k \times (|T| + n_m \times (|V| + |E|)) + \hat{d}))= O(n_l \times k \times n_m \times (|V| + |E|))$, where $n_l$ is the number of targeted LCCs $l_j$ with $n_l=O(\hat{d}^2)$.\footref{complexity}
%\end{proof}

%\newpage
\bibliographystyle{ACM-Reference-Format}
\bibliographystyle{abbrv}
\bibliography{reference}

\end{document}